\documentclass[amssymb,nobibnotes,superscriptaddress,longbibliography,notitlepage,singlecolumn]{revtex4-1}

\usepackage{enumerate,amsmath,amsthm}
\usepackage[dvipdfmx]{graphicx}
\usepackage[dvipdfmx]{color}
\usepackage[colorlinks=true]{hyperref}
\usepackage{comment}
\usepackage{braket}
\usepackage{bm}

\newtheorem{theorem}{Theorem}

\usepackage{algorithm}
\usepackage{algpseudocode}

\usepackage{mathtools}

\newtheorem{definition}{Definition}
\newtheorem{lemma}{Lemma}
\newtheorem{prop}{Proposition}
\newtheorem{prob}{Problem}
\newtheorem{cor}{Corollary}

\allowdisplaybreaks

\begin{document}

\title{How to Map Linear Differential Equations to Schr\"{o}dinger Equations via Carleman and Koopman-von Neumann Embeddings for Quantum Algorithms}

\author{Yuki Ito}
\email{yuki.itoh.osaka@gmail.com}
\affiliation{%
Graduate School of Engineering Science, Osaka University,
1-3 Machikaneyama, Toyonaka, Osaka 560-8531, Japan\\
}%

\author{Yu Tanaka}
\affiliation{
Advanced Research Laboratory, Technology Infrastructure Center, Technology Platform,
Sony Group Corporation, 1-7-1 Konan, Minato-ku, Tokyo, 108-0075, Japan\\
}%

\author{Keisuke Fujii}
\email{fujii@qc.ee.es.osaka-u.ac.jp}
\affiliation{%
Graduate School of Engineering Science, Osaka University,
1-3 Machikaneyama, Toyonaka, Osaka 560-8531, Japan\\
}%
\affiliation{Center for Quantum Information and Quantum Biology, Osaka University, 1-2 Machikaneyama, Toyonaka, Osaka 560-0043, Japan}
\affiliation{RIKEN Center for Quantum Computing (RQC), Hirosawa 2-1, Wako, Saitama 351-0198, Japan}

\date{\today}

\begin{abstract}
Solving linear and nonlinear differential equations with large degrees of freedom is an important task 
for scientific and industrial applications. 
In order to solve such differential equations on a quantum computer, it is necessary to embed 
classical variables into a quantum state. 
While the Carleman and Koopman-von Neumann embeddings have been investigated so far,
the class of problems that can be mapped to the  Schr{\"o}dinger equation is not well understood even for linear differential equations.
In this work, 
we investigate the conditions for linear differential equations
to be mapped to the Schr{\"o}dinger equation and solved on a quantum computer.
Interestingly, we find that these conditions are identical for both Carleman and Koopman-von Neumann embeddings.
We also compute the computational complexity associated with estimating the expected values of an observable. This is done by assuming a state preparation oracle, block encoding of the mapped Hamiltonian via either Carleman or Koopman-von Neumann embedding, and block encoding of the observable using $\mathcal{O}(\log M)$ qubits with $M$ is the mapped system size.
Furthermore, we consider a general classical quadratic Hamiltonian dynamics and find a sufficient condition to map it into the Schr\"{o}dinger equation.
As a special case, this includes the coupled harmonic oscillator model [Babbush et al., \cite{babbush_exponential_2023}]. 
We also find a concrete example that cannot be described as the coupled harmonic oscillator but 
can be mapped to the Schr{\"o}dinger equation in our framework.
These results are important in the construction of quantum algorithms for solving differential equations of large-degree-of-freedom.
\end{abstract}

\maketitle

\section{Introduction}
Differential equations are often used for describing complex phenomena, 
such as the Maxwell equation in electromagnetism, the Navier-Stokes Equation in fluid dynamics, the Magnetohydrodynamics equation in plasma physics, the primitive equation in climate modeling, and 
Lotka-Volterra equation in ecology.
While some differential equations can be solved analytically, many of practical significance cannot. 
To understand the behavior of these complex systems, these equations must be solved numerically. 
However, numerical simulation of such systems with a large number of variables requires substantial computational resources.

Quantum computers are thought to offer powerful computational capability for a certain class of problems such as prime factorization~\cite{shor_polynomial-time_1997}, simulation of quantum many-body system~\cite{reiher_elucidating_2017}, and linear system solver~\cite{harrow_quantum_2009}. 
We aim to harness such an advantage of quantum computers to solve linear or non-linear differential equations in a large number of variables efficiently.
When restricted to linear differential equations,
there have been investigated algorithms for solving sparse linear differential equations using the quantum linear system algorithm and time discretization ~\cite{berry_high-order_2014, berry_quantum_2017}. 
Furthermore, the quantum algorithm \cite{babbush_exponential_2023} to solve coupled harmonic oscillators with exponential speedup has been proposed recently and shown to include a BQP-complete problem.

On the other hand, solving nonlinear differential equations with quantum computers requires embedding the classical variables of the differential equations into quantum states. 
As introduced in Refs. \cite{kowalski_nonlinear_1997, engel_linear_2021}, there are two known ways to realize such embedding.
The first one, Carleman embedding \cite{carleman1932application}, embeds classical variables into the amplitudes of a quantum state directly. 
The second one, Koopman-von Neumann embedding \cite{koopman1931hamiltonian, von-Neumann_1932a, von-Neumann_1932b}, is to embed classical variables into quantum amplitudes through orthogonal polynomials.
In general, these embeddings require infinite variables because Carleman and Koopman-von Neumann embeddings convert
a nonlinear differential equation into a linear one by introducing additional degrees of freedom with respect to nonlinear terms.

One of the advantages of Carleman embedding is the mapping between classical variables and quantum amplitudes is straightforward; each classical variable corresponds to a complex amplitude with an appropriate normalization.
By virtue of this property, if the nonlinear differential equations are sparse, the mapped linear equations will also be sparse.
Based on this mapping, a quantum algorithm to solve a nonlinear differential equation has been proposed ~\cite{liu_efficient_2021}.
Unfortunately, for the Carleman embedding, the matrix defining the mapped linear system is not hermitian in general, 
and hence it cannot be regarded as Schr\"{o}dinger equation.
To handle this,
we have to discretize in the time direction and use quantum linear system algorithms to solve it~\cite{childs_quantum_2017, harrow_quantum_2009}.

On the other hand, 
Koopman-von Neumann embedding \cite{engel_linear_2021, tanaka_quantum_2023} using orthogonal polynomials to map
classical variables to quantum amplitudes provides a linear differential equation that has 
a form of Schr\"{o}dinger equation and hence can be solved by Hamiltonian simulation ~\cite{low_hamiltonian_2019, gilyen_quantum_2019}.
However, as a drawback, the mapped Hamiltonian is not guaranteed to be sparse even if the 
original nonlinear differential equation is sparse. 
In Ref.~\cite{tanaka_quantum_2023}, the authors clarify a sufficient condition on the nonlinear differential equation that the mapped Hamiltonian becomes sparse, and hence can be solved efficiently on a quantum computer.

In the perspective of solving differential equations in a large number of variables with quantum computers, 
it remains unclear what kind of differential equations can be converted into the Schr{\"o}dinger equation either Carleman or Koopman-von Neumann embedding.
Furthermore, it is not well understood how the classes of linear and nonlinear equations that can be solved by Carleman and Koopman-von Neumann embedding differ, 
and in what sense they can be efficiently solved by a quantum computer when the conditions are met.

In this study, we investigate the relation between Carleman and Koopman-von Neumann embedding, specifically focusing on linear differential equations.
We consider general linear differential equations, but not necessarily in the form of the Schr\"{o}dinger equation.
In order to map them to the Schr\"{o}dinger equation via Carleman or Koopman-von Neumann embedding, 
we define a linear transformation to change the variables and clarify the necessary and sufficient condition on the linear differential equations for being able to do that.
Specifically, in the case of the Koopman-von Neumann embedding, 
the dimension of the mapped Schr{\"o}dinger equation becomes infinite even if the original differential equation is linear
because of the nature of the Koopman-von Neumann embedding.
However, if the mapped Hamiltonian preserves the total fock number, 
we can restrict the system into the one-particle subspace and find a finite dimensional Schr{\"o}dinger equation to solve the original problem.
Interestingly,
the necessary and sufficient condition on the linear differential equation to be mapped into the Schr{\"o}dinger equation
is equivalent for both Carleman and Koopman-von Neumann embeddings.
We also argue the computational complexity of the quantum algorithms solving these linear differential equations 
with assuming state preparation oracle, block encoding of the mapped Hamiltonian, and block encoding of the observable, using $\mathcal{O} (\log M)$ qubits, where $M$ is the mapped system size.

Furthermore, we consider a case, where the differential equation is given as 
Hamilton's canonical equation of a quadratic classical Hamiltonian,
which includes coupled harmonic oscillator model~\cite{babbush_exponential_2023} as a special case.
We provide a sufficient condition on classical quadratic Hamiltonian so that 
we can map the corresponding Hamiltonian's canonical equation to the Schr{\"o}dinger equation.
While the coupled harmonic oscillators treated in Ref.~\cite{babbush_exponential_2023} is a special case of the above result,
we also provide a nontrivial example that cannot be described as a coupled harmonic oscillator by the method of Ref.~\cite{babbush_exponential_2023} but 
can be mapped to the Schr{\"o}dinger equation in our framework.
By using the same argument as in Ref.~\cite{babbush_exponential_2023}, we show that
a quantum algorithm obtained via Koopman-von Neumann embedding can solve BQP-complete problems even in the case of linear differential equations.
These results are general in the situation of embedding general linear differential equations into the Schr{\"o}dinger one, and are very important 
for constructing efficient quantum algorithms for linear and nonlinear differential equations in the future.

The rest of the paper is organized as follows.
In Sec. \ref{section:preliminary}, we introduce notations about the mathematical terms of matrices.
In Sec. \ref{section:carleman-linear}, we investigate the necessary and sufficient condition for embedding 
a given differential equation into the Schr{\"o}dinger equation via Carleman embedding.
We also argue about the computational complexity of a related quantum algorithm to solve the linear differential equation.
In Sec. \ref{section:Koopman-von Neumann}, the same argument is expanded for the Koopman-von Neumann embedding.
In Sec. \ref{section:relation-carleman-koopman}, we investigate the relation between Carleman and Koopman-von Neumann embedding,
whether or not one can construct one embedding from another in a constructive way.
In Sec. \ref{section:classical-quadratic-hamiltonian-dynamics}, we consider a special case where the linear differential equation 
is given as a Hamilton's canonical equation with a quadratic classical Hamiltonian.
In Sec. \ref{section:coupled-harmonic-oscillators}, we revisit the result obtained in Ref.~\cite{babbush_exponential_2023}
in our framework.
Sec. \ref{section:conclusion} is devoted to conclusion and discussion.

\section{Preliminary} \label{section:preliminary}
In this section, we define notations in this study.
First, we say $A  \in \mathbb{C}^{N \times N} (N \in \mathbb{Z}_{>0})$ is diagonalizable when there exists a regular matrix $P \in \mathbb{C}^{N \times N}$ such that $P^{-1}AP$ is a diagonal matrix.
Note that $P$ is not necessarily a unitary matrix but a regular matrix in general.
Second, we denote $\sigma_p(A)$ as the set of the eigenvalues of $A  \in \mathbb{C}^{N \times N} (N \in \mathbb{Z}_{>0})$.
Third, we say $A  \in \mathbb{C}^{N \times N} (N \in \mathbb{Z}_{>0})$ is similar to $B \in \mathbb{C}^{N \times N}$ when there exists a regular matrix $P \in \mathbb{C}^{N \times N}$ such that $B=P^{-1}AP$.
Fourth, we define the real and imaginary parts of a matrix $A \in \mathbb{C}^{N \times N^\prime} (N, N^\prime \in \mathbb{Z}_{>0})$ as 
\begin{equation}
    \mathfrak{Re} A = \frac{1}{2} \left(A + A^* \right),
    \hspace{1em}
    \mathfrak{Im} A = \frac{1}{2i} \left(A - A^* \right).
\end{equation}
Fifth, we denote $\mathrm{Im} B$ as the image of $B \in \mathbb{C}^{N \times N^\prime} (N, N^\prime \in \mathbb{Z}_{>0})$.
Finally, we denote $I_N (N \in \mathbb{Z}_{>0})$ as a $N \times N$ identity matrix and $O_{N \times N^\prime} (N, N^\prime \in \mathbb{Z}_{>0})$ as a $N \times N^\prime$ zero matrix.

Next, to extend the scope of linear differential equations that can be mapped to the Schr\"{o}dinger equation, we define the one-way reversible linear transformation as a method for changing the variables as follows.
We say $B \in \mathbb{C}^{M \times N} \ (M,N \in \mathbb{Z}_{>0} \ \mathrm{with} \ M \ge N)$ is a one-way reversible linear transformation when $B$ satisfies $\mathrm{rank} B = N$.
This transformation is a generalization of the technique used in Ref.~\cite{babbush_exponential_2023}.
Next, for a one-way reversible linear transformation $B$, we define the characteristic matrix $C_B$ as $C_B \coloneqq (B^\dagger B)^{-1} B^\dagger$.
The properties of $C_B$ are given as follows:
\begin{lemma} \label{lem:property-CB}
    Let $\mathbb{K}$ be $\mathbb{C} \ \mathrm{or} \ \mathbb{R}$, and $B$ belongs to $\mathbb{K}^{M \times N} \ \mathrm{with} \ \mathrm{rank} B =N$.
    $C_B = (B^\dagger B)^{-1} B^\dagger$ satisfies $C_B B = I_N$ and $C_B |_{(\mathrm{Im} B)^\perp} = 0$.
\end{lemma}

\begin{proof}
    $B$ satisfies $\mathrm{rank} B = N$, so that from rank–nullity theorem we obtain $\dim \mathrm{Ker} B = 0$. Next, $x \in \mathrm{Ker} B^\dagger B$ satisfies $x^\dagger B^\dagger B x = 0$, so $0=x^\dagger B^\dagger B x=\lVert B x \rVert^2$. Now,  $\mathrm{Ker} B = \{0\}$, then $\mathrm{Ker} B^\dagger B = \{0\}$. Therefore, there exists the inverse $(B^\dagger B)^{-1}$ of $B^\dagger B$. This leads the existence of $C_B = (B^\dagger B)^{-1} B^\dagger$. We can calculate like following $C_B B = (B^\dagger B)^{-1} B^\dagger B = I_N$, then we obtain $C_B B = I_N$. 
    Since $(\mathrm{Im} B)^\perp = \mathrm{Ker} B^\dagger$, $C_B = (B^\dagger B)^{-1} B^\dagger$ satisfies $C_B |_{(\mathrm{Im} B)^\perp} = 0$.
\end{proof}

\section{Carleman Embedding for Linear Differential Equations} \label{section:carleman-linear}
We would like to know when a given linear differential equation
can be mapped into the Schr{\"o}dinger equation and then
a quantum computer can be applied to solve it if 
appropriate oracle and/or block encoding are provided.
To do so,
it is necessary to consider a concrete method for how to embed classical variables into quantum states. 
The most straightforward approach would be to regard the vector of classical variables as a quantum state vector with appropriate normalization. 
Since this approach is equivalent to considering only the linear terms of the Carleman embedding (see Appendix \ref{appendix:carleman-embedding} for a review), we will refer to it as Carleman embedding in the linear case as well. 
That is, here we use the terminology, Carleman embedding, as a method to embed classical variables into a quantum state.

Let us start with a general linear differential equation:
\begin{equation*}
    \frac{d}{dt}
    \begin{pmatrix}
        x_1 (t) \\
        \vdots \\
        x_N (t)
    \end{pmatrix}
    =A
    \begin{pmatrix}
        x_1 (t) \\
        \vdots \\
        x_N (t)
    \end{pmatrix},
\end{equation*}
where $A \in \mathbb{R}^{N \times N} (N \in \mathbb{Z}_{>0})$.
Since $A$ is not necessarily anti-hermitian, 
to map the above differential equation to the Schr{\"o}dinger equation,
we define the one-way reversible linear transformations $B$,
which can be regarded as a generalization of the result obtained in Ref.~\cite{babbush_exponential_2023}.
Suppose we have a one-way reversible linear transformation
$B \in \mathbb{C}^{M \times N} \ (M,N \in \mathbb{Z}_{>0} \ \mathrm{with} \ M \ge N)$ and define $\bm{y}(t):= (y_1(t), \dots, y_M(t))^\top = B (x_1(t), \dots, x_N(t))^\top$.
Then we obtain a linear differential equation with respect to $\bm{y}(t)$:
\begin{equation*}
    \frac{d}{dt} 
    \begin{pmatrix}
        y_1(t) \\
        \vdots \\
        y_M(t) 
    \end{pmatrix}
    =
    BAC 
    \begin{pmatrix}
        y_1(t) \\
        \vdots \\
        y_M(t) 
    \end{pmatrix},
\end{equation*}
where $C \in \mathbb{C}^{N \times M}$ satisfies $CB = I_N$.
In this case, we assumed the existence of $C$ satisfying $CB = I_N$ for reversing the transformed solution $\bm{y} (t)$ to the original solution $\bm{x}(t)$, since we want to know the original solution $\bm{x}(t)$ from $\bm{y}(t)$.
We also imposed $\mathrm{rank} B = N$ on $B$ as described in Sec.\ref{section:preliminary} so that there exists $C \in \mathbb{C}^{N \times M}$ such that $CB = I_N$ holds.
As seen in Sec.~\ref{section:preliminary}, 
we can find such a $C$ as a Moore–Penrose pseudo-inverse of $B$, i.e., $C_B = (B^\dagger B)^{-1} B^\dagger$.

The problem we want to solve here is to clarify what is the necessary and sufficient condition on $A$ for $BAC$ to be anti-hermitian,
and hence the linear differential equation with respect to $\bm{y}$ becomes the Schr{\"o}dinger equation.
To do so we define a property, {\it pure imaginary diagonalizable}, of a given matrix $A$ as follows.
\begin{definition}
    $A \in \mathbb{R}^{N \times N}$ is called pure imaginary diagonalizable iff $A$ is diagonalizable and $\sigma_p(A) \subset i \mathbb{R}$.
\end{definition}

We note that if $A$ is pure imaginary diagonalizable then $A$ is similar to a diagonal anti-hermitian matrix, i.e., 
there exists a regular matrix $P \in \mathbb{C}^{N \times N}$ such that $P^{-1}AP$ is the diagonal matrix whose diagonal components are in $i \mathbb{R}$.
We also note that this $P$ is not necessarily a unitary matrix.

\begin{theorem} \label{thm:carleman-schrodinger}
    For a given linear differential equation $\frac{d}{dt} \bm{x}(t) = A \bm{x} (t)$, where $A \in \mathbb{R}^{N \times N}$,
    there exists $(B, C) \in \mathbb{C}^{M \times N} \times \mathbb{C}^{N \times M} \ \mathrm{with} \ \mathrm{rank} B = N \ \mathrm{and} \ CB=I_N$ such that 
    a linear differential equation for $\bm{y}(t):=B\bm{x}(t)$, 
    $\frac{d}{dt} \bm{y}(t)  = BAC \bm{y}(t)$ becomes the Schr\"{o}dinger equation, i.e., $BAC$ is anti-hermitian
    iff $A$ is pure imaginary diagonalizable.
\end{theorem}

\begin{proof}
    We prove this in Appendix \ref{appendix-sub:carleman-thm}.
\end{proof}

Specifically, 
we can take $C$ to be $C_B = (B^{\dag}B)^{-1}B^{\dag}$, i.e., Moore–Penrose pseudo-inverse of $B$, where $C_B B= I_N$ as shown in Lemma~\ref{lem:property-CB}.

We remember how to solve a linear differential equation $\frac{d}{dt} \bm{x}(t) = A \bm{x}(t)$, where $A \in \mathbb{R}^{N \times N}$.
If $A$ satisfies Theorem \ref{thm:carleman-schrodinger}, the solution $\bm{x}$ can be represented as $\bm{x}(t)=P^{-1} \exp(\Lambda t) P \bm{x}(0) $, where the regular matrix $P \in \mathbb{C}^{N \times N}$ such that $PAP^{-1}$ is a diagonal matrix and $\Lambda \coloneqq PAP^{-1}$.
$\Lambda$ is the diagonal matrix whose diagonal components are in $i \mathbb{R}$, so $\exp(\Lambda t)$ is an unitary matrix.
That is,
the one-way reversible linear transformation
change the basis (or classical variables) so that 
the time evolution can be given as a unitary transformation, i.e.,
the Schr\"{o}dinger equation.
Since the new variables $\bm{y}$ can be regarded as complex amplitudes of 
quantum state with an appropriate normalization, 
we call this Carleman embedding of the linear case and $B \in \mathbb{C}^{M \times N}$ is a Carleman transforming matrix of $A$ here when $BAC_B$ is an anti-hermitian matrix.

Note that $A$ itself is not an anti-hermitian matrix in general.
For example, as in Ref.~\cite{babbush_exponential_2023}, we can transform 
Hamilton's canonical equation of the coupled harmonic oscillator into the Schr\"{o}dinger equation as follows:
\begin{equation} \label{eq:coupled-harmonic-oscillators-schrodinger}
    \frac{d}{dt}
    \left(
    B
    \begin{pmatrix}
        x_1(t) \\
        \vdots \\
        x_N(t) \\
        \dot{x}_1(t) \\
        \vdots \\
        \dot{x}_N(t)
    \end{pmatrix}
    \right)
    =
    BAC_{B}
    \left(
    B
    \begin{pmatrix}
        x_1(t) \\
        \vdots \\
        x_N(t) \\
        \dot{x}_1(t) \\
        \vdots \\
        \dot{x}_N(t)
    \end{pmatrix}
    \right),
\end{equation}
where $A \in \mathbb{R}^{2N \times 2N}$ satisfies $\frac{d}{dt} (x_1(t), \dots, x_N(t), \dot{x}_1(t), \dots, \dot{x}_N(t))^\top = A (x_1(t), \dots, x_N(t), \dot{x}_1(t), \dots, \dot{x}_N(t))^\top$,
$B \in \mathbb{C}^{N(N+3)/2 \times 2N}$ satisfies 
\begin{align} \label{eq:quantum-speedup-carleman-B}
B
\begin{pmatrix}
    x_1(t) \\
    \vdots \\
    x_N(t) \\
    \dot{x}_1(t) \\
    \vdots \\
    \dot{x}_N(t)
\end{pmatrix}
=
\begin{pmatrix}
    \sqrt{m_1} \dot{x}_1(t) \\
    \vdots \\
    \sqrt{m_N} \dot{x}_N(t) \\
    i \sqrt{\kappa_{11}} x_1(t) \\
    \vdots \\
    i \sqrt{\kappa_{NN}} x_N(t) \\
    i \sqrt{\kappa_{12}} \left( x_1(t)-x_2(t) \right) \\
    \vdots \\
    i \sqrt{\kappa_{N-1 \, N}} \left( x_{N-1}(t)-x_N(t) \right)
\end{pmatrix},
\end{align}
$C_{B} \in \mathbb{C}^{2N \times N(N+3)/2}$ is defined in Sec. \ref{section:preliminary},
$m_j >0 \ (j \in \{1, \dots, N\})$ is the point positive mass, and $\kappa_{jk} \ (j,k \in \{1, \dots, N\})$ is the spring constant satisfying $\kappa_{jk}=\kappa_{kj} \ge 0$.
$B A C_{B}$ becomes an anti-hermitian matrix, and this $B$ is the example of 
a Carleman transforming matrix of $A$.

Now we consider the following problem to solve a linear differential equation defined by a pure imaginary diagonalizable matrix $A$,
as in Ref.~\cite{babbush_exponential_2023}.
\begin{prob} \label{prob:hamiltonian-simulation-Carleman}
    Let $A$ be pure imaginary diagonalizable and $B$ be a Carleman transforming matrix of $A$.
    Let $\bm{x}(t)$ be a solution of $\frac{d}{dt} \bm{x} (t) = A \bm{x}(t)$ and define the normalized state
    \begin{align}
        \ket{\psi(t)} = \frac{1}{ \lVert B \bm{x} (t) \rVert} B \bm{x}(t).
    \end{align}
    Let observable $O \in \mathbb{C}^{M \times M}$ be a hermitian matrix.
    Assume we are given $(\alpha_{\rm Car}, a_{\rm Car}, 0)$-block encoding $U_{H_{\rm Car}}$ of $H_{\rm Car} \coloneqq iBAC_B$, 
    $(\beta_{\rm Car}, b_{\rm Car}, 0)$-block encoding $U_O$ of $O$,
    and oracle access to a unitary $U_{\rm ini}$ that prepares the initial state, i.e., $U_{\rm ini} \ket{0} \mapsto \ket{\psi(0)}$.
    Given $t \ge 0$ and $\varepsilon > 0$, the goal is to estimate a expectation value that is $\varepsilon$-close to $\braket{\psi(t) | O | \psi(t)}$.
\end{prob}

\begin{theorem} \label{thm:hamiltonian-simulation-Carleman}
    Let $\delta >0$.
    Problem \ref{prob:hamiltonian-simulation-Carleman} can be solved with probability at least $1-\delta$ by a quantum algorithm that makes 
    $\mathcal{O} (\beta_{\rm Car} \log(1/\delta)/\varepsilon)$ uses of $U_O$, $U_O^\dagger$, controlled-$U_O$, controlled-$U_O^\dagger$,
    $U_{\rm ini}$, $U_{\rm ini}^\dagger$ and controlled-$U_{H_{\rm Car}}$ or its inverse, and
    $\mathcal{O} \left( \frac{\beta_{\rm Car} \log(1/\delta)}{\varepsilon} \left(\alpha_{\rm Car} t + \frac{\log (\beta_{\rm Car}/\varepsilon)}{\log (e+\log (\beta_{\rm Car}/\varepsilon)/\alpha_{\rm Car} t)} \right) \right)$ uses of $U_{H_{\rm Car}}$ or its inverse,
    with $\mathcal{O} (\log M)$ qubits.
\end{theorem}

\begin{proof}
    From Theorem 58 and Corollary 60 in Ref.~\cite{gilyen_quantum_2019}, we can implement an $(1,a_{\rm Car}+2, \varepsilon/4 \beta)$-block encoding $V$ of $e^{-it H_{\rm Car}}=e^{BAC_B t}$ with $\mathcal{O} \left(\alpha_{\rm Car} t + \frac{\log (\beta_{\rm Car}/\varepsilon)}{\log (e+\log (\beta_{\rm Car}/\varepsilon)/\alpha_{\rm Car} t)} \right)$ uses of $U_{H_{\rm Car}}$ or its inverse, 
    3 uses of controlled-$U_{H_{\rm Car}}$ or its inverse.
    Then, $\ket{\phi(t)}= (\bra{0}^{\otimes a_{\rm Car}+2} V \ket{0}^{\otimes a_{\rm Car}+2}) U_{\rm ini} \ket{0}$ is $\varepsilon/4\beta_{\rm Car}$-close to $\ket{\psi(t)}$.
    Since $\lVert O \rVert \le \beta_{\rm Car}$ holds by the definition of $U_O$ in Problem \ref{prob:hamiltonian-simulation-Carleman},
    then
    \begin{align}
        &\quad \left\lVert  \braket{\psi(t) | O | \psi(t)} - \beta (\bra{0}^{\otimes b_{\rm Car}} \bra{\phi (t)})  U_O (\ket{0}^{\otimes b_{\rm Car}} \ket{\phi (t)})  \right\rVert \\
        &= \left\lVert \braket{\psi (t) | O | \psi(t)} - \braket{\phi (t) | O | \phi(t)}  \right\rVert \\
        &\le \left\lVert \braket{\psi (t) | O | \psi(t)} - \braket{\psi (t) | O | \phi(t)} \right\rVert + \left\lVert \braket{\psi (t) | O | \phi(t)} - \braket{\phi (t) | O | \phi(t)} \right\rVert \\
        &\le \lVert O \rVert \lVert \ket{\psi (t)} - \ket{\phi(t)} \rVert \lVert \ket{\psi(t)} \rVert + \lVert O \rVert \lVert \ket{\psi (t)} - \ket{\phi(t)} \rVert \lVert \ket{\phi(t)} \rVert \\
        &\le 2 \cdot \beta_{\rm Car} \cdot (\varepsilon/4\beta_{\rm Car}) = \varepsilon/2
    \end{align}
    holds.
    If we can estimate $(\bra{0}^{\otimes b_{\rm Car}} \bra{\phi (t)}) U_O (\ket{0}^{\otimes b_{\rm Car}} \ket{\phi (t)})$ with additive error at most $\varepsilon/2 \beta_{\rm Car}$ and error probability $\delta$, we gain a state that is $\varepsilon$-close to $\braket{\psi(t) | O | \psi(t)}$.
    A method based on high-confidence amplitude estimation \cite{knill_optimal_2007} provides a quantum algorithm for estimating the $(\bra{0}^{\otimes b_{\rm Car}} \bra{\phi (t)}) U_O (\ket{0}^{\otimes b_{\rm Car}} \ket{\phi (t)})$ that makes $\mathcal{O} (\beta_{\rm Car} \log(1/\delta)/\varepsilon)$ uses of $U_O$ and its inverse, $\mathcal{O} (\beta_{\rm Car} \log(1/\delta)/\varepsilon)$ uses of conrolled-$U_O$ and its inverse, and $\mathcal{O} (\beta_{\rm Car} \log(1/\delta)/\varepsilon)$ uses of the quantum circuit that prepares $\ket{\phi(t)}$ or its inverse.
    
    That is, Problem \ref{prob:hamiltonian-simulation-Carleman} can be solved with probability at least $1-\delta$ by a quantum algorithm that makes 
    $\mathcal{O} (\beta_{\rm Car} \log(1/\delta)/\varepsilon)$ uses of $U_O$, $U_O^\dagger$, controlled-$U_O$, controlled-$U_O^\dagger$,
    $U_{\rm ini}$, $U_{\rm ini}^\dagger$ and controlled-$U_{H_{\rm Car}}$ or its inverse, and
    \begin{align*}
    \mathcal{O} \left( \frac{\beta_{\rm Car} \log(1/\delta)}{\varepsilon} \left(\alpha_{\rm Car} t + \frac{\log (\beta_{\rm Car}/\varepsilon)}{\log (e+\log (\beta_{\rm Car}/\varepsilon)/\alpha_{\rm Car} t)} \right) \right)
    \end{align*}
    uses of $U_{H_{\rm Car}}$ or its inverse.
\end{proof}

\section{Koopman-von Neumann Embedding for Linear Differential Equations} \label{section:Koopman-von Neumann}
Next, we apply Koopman-von Neumann embedding for a linear differential equation $\frac{d}{dt} \bm{x}(t) = A \bm{x}(t)$.
Specifically, we employ the Hermite polynomial $\{H_n(x)\}_{n=0}^\infty$ as an orthogonal polynomial sequence for Koopman-von Neumann embedding, for simplicity.
In the Koopman-von Neumann embedding,
we define a position eigenstate:
\begin{align}
    \hat{x}_j  \ket{ \bm{x}   } = x_j \ket{ \bm{x}   },
\end{align}
where $\hat{x}_j$ is a $j$-th position operator, and a mapped Hamiltonian
\begin{align} \label{eq:mapped-hamiltonian}
    \hat{H} = \sum_{j=1}^N \frac{1}{2} \left( \hat{k}_j \left( \sum_{l=1}^N a_{jl}\hat{x}_l \right) +  \left( \sum_{l=1}^N a_{jl}\hat{x}_l \right) \hat{k}_j \right),
\end{align}
where $a_{jl}$ is the $(j,l)$ component of $A$ and 
$\hat{k}_j$ is a $j$-th momentum operator satiftying $[\hat{x}_l, \hat{k}_j]=i \delta_{l,j}$.
The relation between classical variables and quantum states is as follows:
\begin{align}
    \left( \bigotimes_{j=1}^N \bra{n_j} \right) \ket{\bm{x}} = \prod_{j=1}^N w(x_j)^{1/2} H_{n_j} (x_j),
\end{align}
where $\bigotimes_{j=1}^N \ket{n_j}$ is a tensor product of the eigenvector of the number operators and $w(x)$ is the weight function of the Hermite polynomial.
Since the position eigenstate, which is a continuous variable quantum state, is employed,
the mapped Hamiltonian and Schr\"{o}dinger equation $i \frac{d}{dt} \ket{\bm{x}} = \hat{H} \ket{\bm{x}}$ have infinite dimensions
even if the original differential equation is linear and finite dimensional.
Fortunately, when the Hamiltonian $\hat{H}$ preserves the total fock number, $\hat{H}$ can be block diagonalizable under the fock basis, so we focus on the subspace with the total fock number one.
Furthermore, if $\mathrm{div} (A \bm{x})=0$, i.e., $\mathrm{tr} A =0$ holds, that is the weight function is constant over time and the norm of $\ket{\bm{x}}$ does not decay as described in Appendix \ref{appendix:koopman}, the Schr\"{o}dinger equation $i \frac{d}{dt} \ket{\bm{x}} = \hat{H} \ket{\bm{x}}$ restricted in the space with the total fock number one is equivalent to the original linear differential equation $\frac{d}{dt} \bm{x}(t) = A \bm{x}(t)$.

Similarly to Sec. \ref{section:carleman-linear},
we consider the cases where linear differential equations can be transformed,  by a one-way reversible linear transformation, into the Schr\"{o}dinger equation through the Koopman-von Neumann embedding. 
Suppose we have a one-way reversible linear transformation
$B \in \mathbb{C}^{M \times N} \ \mathrm{with} \ \mathrm{rank} B =N \ (M,N \in \mathbb{Z}_{>0} \ \mathrm{with} \ M \ge N)$ and $C \in \mathbb{C}^{N \times M}$ satisfying $CB = I_N$, as in Sec. \ref{section:carleman-linear}.
We define $\bm{y}(t):= (y_1(t), \dots, y_M(t))^\top = B (x_1(t), \dots, x_N(t))^\top$, then obtain a linear differential equation with respect to $\bm{y}(t)$:
\begin{equation*}
    \frac{d}{dt} 
    \begin{pmatrix}
        y_1(t) \\
        \vdots \\
        y_M(t) 
    \end{pmatrix}
    =
    BAC 
    \begin{pmatrix}
        y_1(t) \\
        \vdots \\
        y_M(t) 
    \end{pmatrix}.
\end{equation*}
For the mapped Hamiltonian \eqref{eq:mapped-hamiltonian} being hermitian, we assume $BAC$ is a real matrix.
We also assume $B$ and $C$ are a real matrix so that $BAC$ must be a real matrix, for simplicity.
To preserve the norm of the mapped position state, we have to impose $\mathrm{div} (BAC \bm{y} )=0$, i.e., $\mathrm{tr} (BAC) =0$.
The problem we want to solve here is to clarify what is the necessary and sufficient condition on $A$ which $BAC$ is real anti-hermitian (i.e., real anti-symmetric) and $\mathrm{tr} (BAC) =0$ and also the mapped Hamiltonian preserves the total fock number,
and hence the linear differential equation $\bm{y}$ becomes the Schr{\"o}dinger equation which can be block diagonalizable under the fock basis. 

\begin{theorem} \label{thm:koopman-linear-schrodinger-condition}
    For a given linear differential equation $\frac{d}{dt} \bm{x}(t) = A \bm{x} (t)$ with any initial value $\bm{x} (0) \in \mathbb{R}^N$, where $A \in \mathbb{R}^{N \times N}$,
    there exists $(B, C) \in \mathbb{R}^{M \times N} \times \mathbb{R}^{N \times M} \ \mathrm{with} \ \mathrm{rank} B = N \ \mathrm{and} \ CB=I_N$ such that 
    a linear differential equation for $\bm{y}(t):=B\bm{x}(t)$, 
    $\frac{d}{dt} \bm{y}(t)  = BAC \bm{y}(t)$ satisfies the mapped Hamiltonian preserves the total fock number and $\mathrm{div} (BAC \bm{y})=0$, i.e., $\mathrm{tr} (BAC) =0$ holds
    iff $A$ is pure imaginary diagonalizable.
\end{theorem}

\begin{proof}
    We prove in Appendix \ref{appendix-sub:koopman}.
\end{proof}

Similarly to  Sec. \ref{section:carleman-linear},
we can take $C$ to be $C_B :=(B^{\dag}B)^{-1}B^{\dag}$ for a given 
one-way reversible linear transformation $B$.
We here call such a transformation $B \in \mathbb{R}^{M \times N}$ 
that results in  a real anti-symmetric matrix $BAC_B$ 
via the Koopman-von Neumann embedding, a Koopman-von Neumann transforming matrix of $A$.

Now we consider the following problem to solve a linear differential equation defined by a pure imaginary diagonalizable matrix $A$, using Koopman-von Neumann embedding, as in Sec. \ref{section:carleman-linear}.
\begin{prob} \label{prob:hamiltonian-simulation-Koopman-von Neumann}
    Let $A$ be pure imaginary diagonalizable and $B$ be a Koopman-von Neumann transforming matrix of $A$.
    Let $\bm{x}(t)$ be a solution of $\frac{d}{dt} \bm{x} (t) = A \bm{x}(t)$ and define the normalized state
    \begin{align}
        \ket{\psi(t)} = \frac{1}{ \lVert B \bm{x} (t) \rVert} B \bm{x}(t).
    \end{align}
    Let observable $O \in \mathbb{C}^{M \times M}$ be a hermitian matrix.
    Assume we are given $(\alpha_{\rm Koo}, a_{\rm Koo}, 0)$-block encoding $U_{H_{\rm Koo}}$ of $H_{\rm Koo} \coloneqq iBAC_B$, 
    $(\beta_{\rm Koo}, b_{\rm Koo}, 0)$-block encoding $U_O$ of $O$,
    and oracle access to a unitary $U_{\rm ini}$ that prepares the initial state, i.e., $U_{\rm ini} \ket{0} \mapsto \ket{\psi(0)}$.
    Given $t \ge 0$ and $\varepsilon > 0$, the goal is to estimate a expectation value that is $\varepsilon$-close to $\braket{\psi(t) | O | \psi(t)}$.
\end{prob}

From Corollary \ref{cor:koopman-to-carleman} in Appendix \ref{appendix-sub:koopman}, if $B$ is a Koopman-von Neumann transforming matrix of pure imaginary diagonalizable $A$, then $B$ is a Carleman transforming matrix of $A$.
Thus, we can solve Problem \ref{prob:hamiltonian-simulation-Koopman-von Neumann} in the same way as Problem \ref{prob:hamiltonian-simulation-Carleman}.

\begin{theorem} \label{thm:koopman-simulation-complexity}
    Let $\delta >0$.
    Problem \ref{prob:hamiltonian-simulation-Koopman-von Neumann} can be solved with probability at least $1-\delta$ by a quantum algorithm that makes 
    $\mathcal{O} (\beta_{\rm Koo} \log(1/\delta)/\varepsilon)$ uses of $U_O$, $U_O^\dagger$, controlled-$U_O$, controlled-$U_O^\dagger$,
    $U_{\rm ini}$, $U_{\rm ini}^\dagger$ and controlled-$U_{H_{\rm Koo}}$ or its inverse, and
    $\mathcal{O} \left( \frac{\beta_{\rm Koo} \log(1/\delta)}{\varepsilon} \left(\alpha_{\rm Koo} t + \frac{\log (\beta_{\rm Koo}/\varepsilon)}{\log (e+\log (\beta_{\rm Koo}/\varepsilon)/\alpha_{\rm Koo} t)} \right) \right)$ uses of $U_{H_{\rm Koo}}$ or its inverse,
    with $\mathcal{O} (\log M)$ qubits.
\end{theorem}

We note that when considering the case of $d=2$ in the quantum solvable ODE of Ref.~\cite{tanaka_quantum_2023}, 
quantum solvable ODE becomes a linear differential equation and $H_{\rm Koo}$ becomes a sparse Hamiltonian.
We also note that Ref.~\cite{tanaka_quantum_2023} provides the computational complexity required for the construction of $H_{\rm Koo}$, so by combining it with Theorem \ref{thm:koopman-simulation-complexity}, 
we can obtain a more detailed computational complexity.

So far, we have considered the cases where linear differential equations can be transformed,  by a one-way reversible linear transformation (i.e., Carleman or Koopman-von Neumann transforming matrices), into the Schr\"{o}dinger equation through the Carleman or Koopman-von Neumann embedding. 
Interestingly, as a result of the discussion above, we found that in both cases, the condition for reducing to the Schr\"{o}dinger equation is the same: $A$ must be pure imaginary diagonalizable. 

In contrast to the existing approach to solve nonlinear differential equations using the linear system solver~\cite{liu_efficient_2021},
it is unique that the present approach does not require any information about the condition number or discretization of time of the given linear differential equations. 
Once the Carleman transforming matrix or Koopman-von Neumann transforming matrix is established, and a block encoding of the mapped Hamiltonian is given,
the present approach allows us to simulate the linear differential equation by Hamiltonian simulation.
On the other hand, when the original linear differential equations are sparse, the sparsity is preserved when using the linear system solver to solve linear differential equations ~\cite{liu_efficient_2021}.
However, in our approach, 
it is not guaranteed that the obtained differential equations (Schr{\"o}dinger equation) will also exhibit sparsity, presenting a potential drawback.
In the next section, we will explore the relation between the Carleman and Koopman-von Neumann embeddings and see the Carleman and Koopman-von Neumann transforming matrices can be constructed from each other.

\section{The Relation between Carleman and Koopman-von Neumann Embedding} \label{section:relation-carleman-koopman}
Next, we investigate how a Carleman transforming matrix and a Koopman-von Neumann transforming matrix can be constructed from each other.
Corollary \ref{cor:koopman-to-carleman} in Appendix \ref{appendix-sub:koopman} tells us one of the ways to transform a Koopman-von Neumann transforming matrix into a Carleman transforming matrix.
Then, we search for the conversely way to transform a Carleman transforming matrix to a Koopman-von Neumann transforming matrix.

\begin{theorem} \label{thm:carleman-to-koopman}
    Let $A \in \mathbb{R}^{N \times N}$ be pure imaginary diagonalizable, and $B \in \mathbb{C}^{M \times N}$ be a Carleman transforming matrix of $A$.
    Then,
    $
    B^\prime \coloneqq 
    \begin{pmatrix}
        \mathfrak{Re} B \\
        \mathfrak{Im} B
    \end{pmatrix}
    \in \mathbb{R}^{2M \times N}$
    is a Koopman-von Neumann transforming matrix. 
\end{theorem}

\begin{proof}
    We prove in Appendix \ref{appendix-sub:real-imaginary}.
\end{proof}

Theorem \ref{thm:carleman-to-koopman} also tells us Koopman-von Neumann embedding is a ``decomplixification'' (converse of the complexification) of Carleman embedding in some sense.

Next, we investigate the sparsity of the matrices obtained via the Carleman transforming matrix and the Koopman-von Neumann transforming matrix, when they are applied to a given pure imaginary diagonalizable matrix.
Corollary \ref{cor:koopman-to-carleman} shows the Koopman-von Neumann transforming matrix is also the Carleman transforming matrix.
Considering them as the same,
the sparsity of the Hamiltonian transformed by the Carleman transforming matrix is preserved when compared to the Hamiltonian transformed by the Koopman-von Neumann transforming matrix.
On the other hand, when constructing the Koopman-von Neumann transforming matrix from the Carleman transforming matrix using the method outlined in Theorem \ref{thm:carleman-to-koopman}, 
the sparsity of the Hamiltonian transformed by the Carleman transforming matrix is inherited by the Hamiltonian transformed by the Koopman-von Neumann transforming matrix under some assumptions, as stated in Proposition \ref{prop:sparsity-Carleman-to-Koopman-von Neumann}.

\begin{prop} \label{prop:sparsity-Carleman-to-Koopman-von Neumann}
    Let $A \in \mathbb{R}^{N \times N}$ be pure imaginary diagonalizable, and $B \in \mathbb{C}^{M \times N}$ be a Carleman transforming matrix of $A$.
    If $B^\dagger B \in \mathbb{R}^{N \times N}$ holds, $BAC_B$ is $s$-sparse and $BAC_B^*$ is $s^\prime$-sparse, then $B^\prime AC_{B^\prime}$ is $2(s+s^\prime)$-sparse, where $
    B^\prime \coloneqq 
    \begin{pmatrix}
        \mathfrak{Re} B \\
        \mathfrak{Im} B
    \end{pmatrix}
    \in \mathbb{R}^{2M \times N}$.
\end{prop}

\begin{proof}
    Since $BAC_B$ is $s$-sparse, $\mathfrak{Re} (BAC_B)$ and $\mathfrak{Im} (BAC_B)$ are $s$-sparse.
    Similarly, since $BAC_B^*$ is $s^\prime$-sparse, $\mathfrak{Re} (BAC_B^*)$ and $\mathfrak{Im} (BAC_B^*)$ are $s^\prime$-sparse.
    Since $A$ and $B^\dagger B$ are real matrices and $C_B = (B^\dagger B)^{-1} B^\dagger$ by Lemma \ref{lem:property-CB}, 
    \begin{align}
        \mathfrak{Re} (BAC_B) &= (\mathfrak{Re} B) A (B^\dagger B)^{-1} (\mathfrak{Re} B)^\top + (\mathfrak{Im} B) A (B^\dagger B)^{-1} (\mathfrak{Im} B)^\top \\
        \mathfrak{Im} (BAC_B) &= (\mathfrak{Im} B) A (B^\dagger B)^{-1} (\mathfrak{Re} B)^\top - (\mathfrak{Re} B) A (B^\dagger B)^{-1} (\mathfrak{Im} B)^\top \\
        \mathfrak{Re} (BAC_B^*) &= (\mathfrak{Re} B) A (B^\dagger B)^{-1} (\mathfrak{Re} B)^\top - (\mathfrak{Im} B) A (B^\dagger B)^{-1} (\mathfrak{Im} B)^\top \\
        \mathfrak{Im} (BAC_B^*) &= (\mathfrak{Im} B) A (B^\dagger B)^{-1} (\mathfrak{Re} B)^\top + (\mathfrak{Re} B) A (B^\dagger B)^{-1} (\mathfrak{Im} B)^\top
    \end{align}
    holds.
    Thus,
    \begin{align}
        (\mathfrak{Re} B) A (B^\dagger B)^{-1} (\mathfrak{Re} B)^\top &= \frac{1}{2} \left(\mathfrak{Re} (BAC_B) + \mathfrak{Re} (BAC_B^*) \right) \\
        (\mathfrak{Im} B) A (B^\dagger B)^{-1} (\mathfrak{Im} B)^\top &= \frac{1}{2} \left(\mathfrak{Re} (BAC_B) - \mathfrak{Re} (BAC_B^*) \right) \\
        (\mathfrak{Im} B) A (B^\dagger B)^{-1} (\mathfrak{Re} B)^\top &= \frac{1}{2} \left(\mathfrak{Im} (BAC_B^*) + \mathfrak{Im} (BAC_B) \right) \\
        (\mathfrak{Re} B) A (B^\dagger B)^{-1} (\mathfrak{Im} B)^\top &= \frac{1}{2} \left(\mathfrak{Im} (BAC_B^*) - \mathfrak{Im} (BAC_B) \right)
    \end{align}
    holds.
    
    Then, $(\mathfrak{Re} B) A (B^\dagger B)^{-1} (\mathfrak{Re} B)^\top, (\mathfrak{Im} B) A (B^\dagger B)^{-1} (\mathfrak{Im} B)^\top, (\mathfrak{Im} B) A (B^\dagger B)^{-1} (\mathfrak{Re} B)^\top$ and $(\mathfrak{Re} B) A (B^\dagger B)^{-1} (\mathfrak{Im} B)^\top$ are $(s+s^\prime)$-sparse.
    $B^\prime AC_{B^\prime}$ is 
    \begin{align}
        B^\prime AC_{B^\prime} &= B^\prime A(B^{\prime \dagger} B^\prime)^{-1} B^{\prime \dagger} \\
        &=
        \begin{pmatrix}
            \mathfrak{Re} B \\
            \mathfrak{Im} B
        \end{pmatrix}
        A (B^\dagger B)^{-1}
        \begin{pmatrix}
            (\mathfrak{Re} B)^\top &(\mathfrak{Im} B)^\top
        \end{pmatrix} \\
        &=
        \begin{pmatrix}
            (\mathfrak{Re} B) A (B^\dagger B)^{-1} (\mathfrak{Re} B)^\top &(\mathfrak{Re} B) A (B^\dagger B)^{-1} (\mathfrak{Im} B)^\top \\
            (\mathfrak{Im} B) A (B^\dagger B)^{-1} (\mathfrak{Re} B)^\top &(\mathfrak{Im} B) A (B^\dagger B)^{-1} (\mathfrak{Im} B)^\top
        \end{pmatrix},
    \end{align}
    so $B^\prime AC_{B^\prime}$ is $2(s+s^\prime)$-sparse.
\end{proof}

Proposition \ref{prop:sparsity-Carleman-to-Koopman-von Neumann} shows 
there is the possibility that 
the way to solve a BQP-complete problem efficiently with Carleman embedding and sparse Hamiltonian simulation converts easily to that with Koopman-von Neumann embedding and sparse Hamiltonian simulation.

\section{Classical Quadratic Hamiltonian Dynamics} \label{section:classical-quadratic-hamiltonian-dynamics}
Next, we focus on the case where the differential equation is given as Hamilton's canonical equation of a quadratic classical Hamiltonian, using the above results.
The authors of Ref.~\cite{babbush_exponential_2023} regard the following equation as the generalized form:
\begin{align} \label{eq:speedup-start}
    M \frac{d^2}{dt^2} \bm{x} = -F \bm{x},
\end{align}
where $M$ is a positive matrix, $F$ is a positive-semidefinite matrix, and $\bm{x}=(x_1, \dots, x_N)^\top$.
Eq. \eqref{eq:speedup-start}'s Hamiltonian $H(q_1, \dots, q_N, p_1, \dots, p_N)$, where $q_j=x_j, p_j=m_j \dot{x}_j\ (j \in \{1, \dots, N\})$ and $m_j > 0 \ (j \in \{1, \dots, N\})$ is the point positive mass, is
\begin{align} \label{eq:quantum-speedup-hamiltonian}
    H(q_1, \dots, q_N, p_1, \dots, p_N) &=
    \begin{pmatrix}
        q_1 &\hdots &q_N &p_1 &\hdots &p_N
    \end{pmatrix}
    \begin{pmatrix}
        H_1 &O_{N \times N} \\
        O_{N \times N} &H_2
    \end{pmatrix}
    \begin{pmatrix}
        q_1 \\
        \vdots \\
        q_N \\
        p_1 \\
        \vdots \\
        p_N
    \end{pmatrix},
\end{align}
where $H_1 \ge 0$ and $H_2>0$.
This Hamiltonian's canonical equation can be mapped into the Schr\"{o}dinger equation according to Ref.~\cite{babbush_exponential_2023}.
Since the Hamiltonian $H(q_1, \dots, q_N, p_1, \dots, p_N)$ is imposed strong constraints,
it is expected that more general Hamiltonian's canonical equation can be mapped to the Schr\"{o}dinger equation.
Indeed, we show the quadratic Hamiltonian whose quadratic form is positive or negative can be mapped to the Schr\"{o}dinger equation.
Furthermore, we provide a nontrivial example that cannot be treated by the way in Ref.~\cite{babbush_exponential_2023} but can be mapped to the Schr{\"o}dinger equation in our framework.

\subsection{General Case} \label{subsec:classical-quadratic-hamiltonian-dynamics-general}
We investigate a sufficient condition on classical quadratic Hamiltonian so that 
we can map the corresponding Hamiltonian's canonical equation to the Schr{\"o}dinger equation.
Assume that the Hamiltonian $H(q_1, \dots, q_N, p_1, \dots, p_N)$ is given that is a quadratic form of generalized position and momentum $q_1, \dots, q_N, p_1, \dots, p_N \ (N \in \mathbb{Z}_{>0})$ 
as follows:
\begin{equation}
    H(q_1, \dots, q_N, p_1, \dots, p_N)
    =
    \begin{pmatrix}
        q_1 &\dots &q_N &p_1 &\dots &p_N
    \end{pmatrix}
    \tilde{H}
    \begin{pmatrix}
        q_1 \\
        \vdots \\
        q_N \\
        p_1 \\
        \vdots \\
        p_N
    \end{pmatrix},
\end{equation}
where $\tilde{H} \in \mathbb{R}^{2N \times 2N}$ is taken as a symmetric matrix without loss of generality.
Hamilton's canonical equation of $H(q_1, \dots, q_N, p_1, \dots, p_N)$ is
\begin{equation} \label{eq:hamilton-canonical-linear}
    \frac{d}{dt}
    \begin{pmatrix}
        q_1 \\
        \vdots \\
        q_N \\
        p_1 \\
        \vdots \\
        p_N
    \end{pmatrix}
    =
    A
    \begin{pmatrix}
        q_1 \\
        \vdots \\
        q_N \\
        p_1 \\
        \vdots \\
        p_N
    \end{pmatrix},
\end{equation}
where $A \coloneqq 
2
\begin{pmatrix}
    O_{N \times N} &I_N \\
    -I_N &O_{N \times N}
\end{pmatrix}
\tilde{H}
$.

Now, we assume that $\tilde{H} > 0$.
According to Lemma \ref{lem:real-positive} in Appendix \ref{appendix-sub:koopman}, there exists $B \in \mathbb{R}^{M \times 2N} \ \mathrm{with} \ \mathrm{rank} B =2N$ such that $\tilde{H} = B^\top B$.
Eq. \eqref{eq:hamilton-canonical-linear} become
\begin{equation} \label{eq:hamiton-caonical-linear-B}
    \begin{split}
        \frac{d}{dt}
        B
        \begin{pmatrix}
            q_1 \\
            \vdots \\
            q_N \\
            p_1 \\
            \vdots \\
            p_N
        \end{pmatrix}
        &=
        B
        A
        C_B B
        \begin{pmatrix}
            q_1 \\
            \vdots \\
            q_N \\
            p_1 \\
            \vdots \\
            p_N
        \end{pmatrix}\\
        &=
        \left(
        2
        B
        \begin{pmatrix}
            O_{N \times N} &I_N \\
            -I_N &O_{N \times N}
        \end{pmatrix}
        B^\top
        \right)
        B
        \begin{pmatrix}
            q_1 \\
            \vdots \\
            q_N \\
            p_1 \\
            \vdots \\
            p_N
        \end{pmatrix},
    \end{split}
\end{equation}
where $C_B = (B^\dagger B)^{-1} B^\dagger = (B^\top B)^{-1} B^\top = \tilde{H}^{-1} B^\top$ from Lemma \ref{lem:property-CB}.
Since
\begin{align}
2
B
\begin{pmatrix}
    O_{N \times N} &I_N \\
    -I_N &O_{N \times N}
\end{pmatrix}
B^\top
\end{align}
is a real anti-symmetric matrix, thus the following Proposition \ref{prop:positive-hamiltonian-carleman-koopman} holds by Proposition \ref{prop:koopman-preserving} in Appendix \ref{appendix-sub:koopman}.

\begin{prop} \label{prop:positive-hamiltonian-carleman-koopman}
    Under the above setting, if $\tilde{H} > 0$ holds, 
    $A \coloneqq 
    2
    \begin{pmatrix}
        O_{N \times N} &I_N \\
        -I_N &O_{N \times N}
    \end{pmatrix}
    \tilde{H}
    $ becomes a pure imaginary diagonalizable matrix.
\end{prop}

Next, we assume that $-\tilde{H} > 0$.
From Lemma \ref{lem:real-positive}, there exists $B \in \mathbb{R}^{M \times 2N} \ \mathrm{with} \ \mathrm{rank} B =2N$ such that $-\tilde{H} = B^\top B$.
Eq. \eqref{eq:hamilton-canonical-linear} become
\begin{equation} \label{eq:hamiton-caonical-linear-negative}
    \begin{split}
        \frac{d}{dt}
        (-B)
        \begin{pmatrix}
            q_1 \\
            \vdots \\
            q_N \\
            p_1 \\
            \vdots \\
            p_N
        \end{pmatrix}
        &=
        (-B)
        \left(
        2
        \begin{pmatrix}
            O_{N \times N} &I_N \\
            -I_N &O_{N \times N}
        \end{pmatrix}
        \tilde{H} 
        \right)
        C_{-B} (-B)
        \begin{pmatrix}
            q_1 \\
            \vdots \\
            q_N \\
            p_1 \\
            \vdots \\
            p_N
        \end{pmatrix}\\
        &=
        \left(
        2
        (-B)
        \begin{pmatrix}
            O_{N \times N} &-I_N \\
            I_N &O_{N \times N}
        \end{pmatrix}
        (-B)^\top
        \right)
        (-B)
        \begin{pmatrix}
            q_1 \\
            \vdots \\
            q_N \\
            p_1 \\
            \vdots \\
            p_N
        \end{pmatrix},
    \end{split}
\end{equation}
where $C_{-B} = \left((-B)^\dagger (-B)\right)^{-1} (-B)^\dagger = \left((-B)^\top (-B)\right)^{-1} (-B)^\top = \tilde{H}^{-1} B^\top$ from Lemma \ref{lem:property-CB}.
Since
\begin{align}
2
(-B)
\begin{pmatrix}
    O_{N \times N} &-I_N \\
    I_N &O_{N \times N}
\end{pmatrix}
(-B)^\top
\end{align}
is a real anti-symmetric matrix, the following Proposition \ref{prop:negative-hamiltonian-carleman-koopman} holds by Proposition \ref{prop:koopman-preserving} in Appendix \ref{appendix-sub:koopman}.

\begin{prop} \label{prop:negative-hamiltonian-carleman-koopman}
    Under the above setting, if $-\tilde{H} > 0$ holds, 
    $A \coloneqq 
    2
    \begin{pmatrix}
        O_{N \times N} &I_N \\
        -I_N &O_{N \times N}
    \end{pmatrix}
    \tilde{H}
    $ becomes a pure imaginary diagonalizable matrix.
\end{prop}

According to Proposition \ref{prop:positive-hamiltonian-carleman-koopman} and \ref{prop:negative-hamiltonian-carleman-koopman}, $\tilde{H}$ being a positive or negative matrix is the sufficient condition on classical quadratic Hamiltonian so that 
we can map the corresponding Hamiltonian's canonical equation to the Schr{\"o}dinger equation.

\subsection{Concrete Example} \label{subsec:classical-quadratic-hamiltonian-dynamics-example}
We provide an example that cannot be dealt with the method in Ref.~\cite{babbush_exponential_2023} but can be in the present framework.
Let Hamiltonian $H(q_1, q_2, p_1, p_2)$ be
\begin{align} \label{eq:hamiltonian-example}
    H(q_1, q_2, p_1, p_2) &= q_1^2 + 2 q_2^2 + p_1^2 + 2 p_2^2 + 2 q_1 p_2 + 2q_2 p_1 \\
    &=
    \begin{pmatrix}
        q_1 &q_2 &p_1 &p_2
    \end{pmatrix}
    \tilde{H}
    \begin{pmatrix}
        q_1 \\
        q_2 \\
        p_1 \\
        p_2
    \end{pmatrix},
\end{align}
where $\tilde{H}
=
\begin{pmatrix}
    1 &0 &0 &1 \\
    0 &2 &1 &0 \\
    0 &1 &1 &0 \\
    1 &0 &0 &2
\end{pmatrix}
$.
From the Hamilton's canonical equation of $H(q_1, q_2, p_1, p_2)$, we obtain
\begin{equation} \label{eq:second-order-ode-position}
    \frac{d^2}{dt^2} 
    \begin{pmatrix}
        q_1 \\
        q_2 \\
        p_1 \\
        p_2
    \end{pmatrix}
    =
    -
    \begin{pmatrix}
        0 &0 &0 &-4 \\
        0 &12 &4 &0 \\
        0 &-4 &0 &0 \\
        4 &0 &0 &12
    \end{pmatrix}
    \begin{pmatrix}
        q_1 \\
        q_2 \\
        p_1 \\
        p_2
    \end{pmatrix}
    \eqqcolon
    -A^{\prime}
    \begin{pmatrix}
        q_1 \\
        q_2 \\
        p_1 \\
        p_2
    \end{pmatrix}.
\end{equation}
Since $A^{\prime}$ in Eq. \eqref{eq:second-order-ode-position} is not a real symmetric positive semidefinite matrix, 
Eq. \eqref{eq:second-order-ode-position} cannot be applied to Theorem 4 in Ref.~\cite{babbush_exponential_2023} that generalizes the results in Ref.~\cite{babbush_exponential_2023}.
Thus, Eq. \eqref{eq:second-order-ode-position} cannot be dealt with the method in Ref.~\cite{babbush_exponential_2023}.

On the other hand, since $\tilde{H}$ is a positive matrix, the Hamilton's canonical equation of $H(q_1, q_2, p_1, p_2)$ is mapped into Schr\"{o}dinger equation under an appropriate Carleman transforming matrix or Koopman-von Neumann transforming matrix from Proposition \ref{prop:positive-hamiltonian-carleman-koopman}.
Indeed,
we can take a Koopman-von Neumann transforming matrix $B$ as
\begin{align}
    B= \frac{1}{\sqrt{5}}
    \begin{pmatrix}
        2 &0 &0 &1 \\
        0 &3 &1 &0 \\
        0 &1 &2 &0 \\
        1 &0 &0 &3
    \end{pmatrix}
\end{align}
because $B^\top B A$ is a real anti-symmetric matrix, and Proposition \ref{prop:koopman-preserving} and \ref{prop:tBBA-anti-symmetric} in Appendix \ref{appendix-sub:koopman} holds.
Then the mapped Hamiltonian becomes
\begin{align}
    2i
    B
    \begin{pmatrix}
        O_{2 \times 2} &I_2 \\
        -I_2 &O_{2 \times 2}
    \end{pmatrix}
    B^\top
    =
    \frac{2i}{5}
    \begin{pmatrix}
        0 &-1 &3 &0 \\
        1 &0 &0 &8 \\
        -3 &0 &0 &1 \\
        0 &-8 &-1 &0
    \end{pmatrix}.
\end{align}
Therefore, the Hamilton's canonical equation of the above $H(q_1, q_2, p_1, p_2)$ is a nontrivial example that cannot be treated in the method of Ref.~\cite{babbush_exponential_2023}
but can be mapped to the Schr{\"o}dinger equation in our framework.

\section{Coupled Harmonic Oscillators} \label{section:coupled-harmonic-oscillators}
In this section, we revisit the coupled harmonic oscillators described in Ref.~\cite{babbush_exponential_2023}.
We demonstrate a sufficient condition to map the Hamilton's canonical equation of the coupled harmonic oscillators to the Schr{\"o}dinger equation.
We also show the BQP-complete problem in Ref.~\cite{babbush_exponential_2023} belongs to the problem class that can be solved efficiently with Koopman-von Neumann embedding.

\subsection{Carleman and Koopman-von Neumann Embedding for Coupled Harmonic Oscillators} \label{subsec:carleman-koopman-coupled-harmonic-oscillators}
We investigate a sufficient condition on coupled harmonic oscillators so that 
we can map the corresponding Hamiltonian's canonical equation to the Schr{\"o}dinger equation.
Coupled harmonic oscillators in Ref.~\cite{babbush_exponential_2023} are given by
\begin{equation} \label{eq:coupled-harmonic-oscillators-Newton}
    m_j \ddot{x}_j(t) = 
    \sum_{\substack{k=1 \\ k \neq j}}^N \kappa_{jk} \left( x_k(t)-x_j(t) \right) - \kappa_{jj} x_j(t)
    \hspace{1em}
    (j \in \{1, \dots, N\}),
\end{equation}
where $m_j$ is the point positive mass, and $\kappa_{jk}$ is the spring constant satisfying $\kappa_{jk}=\kappa_{kj} \ge 0$.
Eq. \eqref{eq:coupled-harmonic-oscillators-Newton} has the Hamiltonian
\begin{equation}
    H(q_1, \dots, q_N, p_1, \dots, p_N) 
    = \sum_{j=1}^N \frac{1}{2} \kappa_{jj} q_j^2 
    + \sum_{1 \le j < k \le N} \frac{1}{2} \kappa_{jk} (q_j - q_k)^2 
    + \sum_{j=1}^N \frac{1}{2 m_j} p_j^2 ,
\end{equation}
where $q_j=x_j,  p_j=m_j \dot{x}_j$.
The Hamiltonian can be represented as
\begin{equation}
    \begin{split}
        &\quad H(q_1, \dots, q_N, p_1, \dots, p_N) \\
        &=
        \frac{1}{2}
        \begin{pmatrix}
            q_1 &\hdots &\hdots &q_N &p_1 &\hdots &\hdots &p_N
        \end{pmatrix} \\
        &\quad
        \begin{pmatrix}
            \sum_{j=1}^N \kappa_{1j} &-\kappa_{12} &\hdots &-\kappa_{1N} 
            &0 &\hdots &\hdots &0 \\
            -\kappa_{21} &\ddots &\ddots &\vdots 
            &\vdots &\ddots &  &\vdots \\
            \vdots &\ddots &\ddots &-\kappa_{N-1 \, N} 
            &\vdots & &\ddots  &\vdots \\
            -\kappa_{N1} &\hdots &-\kappa_{N \, N-1} &\sum_{j=1}^N \kappa_{Nj} 
            &0 &\hdots &\hdots  &0 \\
            0 &\hdots &\hdots &0 &1/m_1 &0 &\hdots &0 \\
            \vdots &\ddots &  &\vdots &0 &\ddots &\ddots &\vdots \\
            \vdots & &\ddots  &\vdots &\vdots &\ddots &\ddots &0 \\
            0 &\hdots &\hdots  &0 &0 &\hdots &0 &1/m_N
        \end{pmatrix}
        \begin{pmatrix}
            q_1 \\
            \vdots \\
            \vdots \\
            q_N \\
            p_1 \\
            \vdots \\
            \vdots \\
            p_N
        \end{pmatrix}.
    \end{split}
\end{equation}
Hamilton's canonical equation of $H(q_1, \dots, q_N, p_1, \dots, p_N)$ becomes
\begin{equation} \label{eq:hamilton-canonical-oscillator}
    \begin{split}
        \frac{d}{dt}
        \begin{pmatrix}
            q_1 \\
            \vdots \\
            \vdots \\
            q_N \\
            p_1 \\
            \vdots \\
            \vdots \\
            p_N
        \end{pmatrix}
        &=
        \begin{pmatrix}
            0 &\hdots &\hdots &0 &1/m_1 &0 &\hdots &0 \\
            \vdots &\ddots &  &\vdots &0 &\ddots &\ddots &\vdots \\
            \vdots & &\ddots  &\vdots &\vdots &\ddots &\ddots &0 \\
            0 &\hdots &\hdots  &0 &0 &\hdots &0 &1/m_N \\
            -\sum_{j=1}^N \kappa_{1j} &\kappa_{12} &\hdots &\kappa_{1N} 
            &0 &\hdots &\hdots &0 \\
            \kappa_{21} &\ddots &\ddots &\vdots 
            &\vdots &\ddots &  &\vdots \\
            \vdots &\ddots &\ddots &\kappa_{N-1 \, N} 
            &\vdots & &\ddots  &\vdots \\
            \kappa_{N1} &\hdots &\kappa_{N \, N-1} &-\sum_{j=1}^N \kappa_{Nj} 
            &0 &\hdots &\hdots  &0
        \end{pmatrix}
        \begin{pmatrix}
            q_1 \\
            \vdots \\
            \vdots \\
            q_N \\
            p_1 \\
            \vdots \\
            \vdots \\
            p_N
        \end{pmatrix} \\
        &\eqqcolon
        A
        \begin{pmatrix}
            q_1 \\
            \vdots \\
            q_N \\
            p_1 \\
            \vdots \\
            p_N
        \end{pmatrix}.
    \end{split}
\end{equation}
Since 
\begin{align}
    \begin{pmatrix}
        1/m_1 &0 &\hdots &0 \\
        0 &\ddots &\ddots &\vdots \\
        \vdots &\ddots &\ddots &0 \\
        0 &\hdots &0 &1/m_N    
    \end{pmatrix}
\end{align}
is a positive matrix,
according to Proposition \ref{prop:positive-hamiltonian-carleman-koopman}, if 
\begin{align}
    \begin{pmatrix}
        \sum_{j=1}^N \kappa_{1j} &-\kappa_{12} &\hdots &-\kappa_{1N} \\
        -\kappa_{21} &\ddots &\ddots &\vdots \\
        \vdots &\ddots &\ddots &-\kappa_{N-1 \, N} \\
        -\kappa_{N1} &\hdots &-\kappa_{N \, N-1} &\sum_{j=1}^N \kappa_{Nj}
    \end{pmatrix}
\end{align}
is a positive matrix,
then $\tilde{H}$ is a positive matrix and the Hamilton's canonical equation of $H(q_1, \dots, q_N, p_1, \dots, p_N)$ will be mapped to the Schr\"{o}dinger equation.
Therefore, 
\begin{align}
\begin{pmatrix}
    \sum_{j=1}^N \kappa_{1j} &-\kappa_{12} &\hdots &-\kappa_{1N} \\
    -\kappa_{21} &\ddots &\ddots &\vdots \\
    \vdots &\ddots &\ddots &-\kappa_{N-1 \, N} \\
    -\kappa_{N1} &\hdots &-\kappa_{N \, N-1} &\sum_{j=1}^N \kappa_{Nj}
\end{pmatrix}
\end{align}
being a positive matrix is the sufficient condition on coupled harmonic oscillators so that 
we can map the corresponding Hamiltonian's canonical equation to the Schr{\"o}dinger equation.

\subsection{BQP-complete Problem that can be Solved by Koopman-von Neumann Embedding} \label{subsec:BQP-completeness-Koopman-von Neumann}
We will show there is a BQP-complete problem in the class that can be solved efficiently by Koopman-von Neumann embedding.
The authors of Ref.~\cite{babbush_exponential_2023} demonstrate that 
the problem of estimating the ratio of the kinetic energy of a specific coupled harmonic oscillator to the total energy of the system (CHO problem) can be solved efficiently using Carleman embedding and Hamiltonian simulation.
They also prove that the CHO problem includes a BQP-complete problem.
From Theorem \ref{thm:carleman-schrodinger} and \ref{thm:koopman-linear-schrodinger-condition}, 
the class of linear differential equations that can be mapped to Schr\"{o}dinger equation by Carleman embedding is identical to that class by Koopman-von Neumann embedding.
Therefore, it appears that the CHO problem can be solved efficiently with Koopman-von Neumann embedding.
However, it is unclear whether the CHO problem can be resolved efficiently using Koopman-von Neumann embedding
because it is not proved that the computational complexity of solving the CHO problem with Koopman-von Neumann embedding is equivalent to that with Carleman embedding.
To show the CHO problem can be solved efficiently using Koopman-von Neumann embedding, 
we need to prove the mapped Hamiltonian via Koopman-von Neumann embedding $H_{\rm Koo}$ can be constructed efficiently.
Since Ref.~\cite{babbush_exponential_2023} shows $H_{\rm Car}$ can be constructed efficiently,
we only have to demonstrate the existence of the unitary matrix $U$ that satisfies $H_{\rm Koo} = U H_{\rm Car} U^\dagger$ and can be constructed efficiently.

The Carleman transforming matrix $B \in \mathbb{C}^{N(N+3)/2 \times 2N}$ of $A$ in Ref.~\cite{babbush_exponential_2023} satisfies Eq. \eqref{eq:coupled-harmonic-oscillators-schrodinger} and \eqref{eq:quantum-speedup-carleman-B}.
Then, we construct a Koopman-von Neumann transforming matrix from $B$.
From Theorem \ref{thm:carleman-to-koopman}, 
$B^\prime \coloneqq 
\begin{pmatrix}
    \mathfrak{Re} B \\
    \mathfrak{Im} B
\end{pmatrix}$
is a Koopman-von Neumann transforming matrix of $A$.
Let $B^{\prime \prime} \in \mathbb{R}^{N(N+3)/2 \times N}$ be the matrix obtained by removing columns with all-zero elements from $B^\prime$,
satisfying
\begin{align}
B^{\prime \prime}
\begin{pmatrix}
    x_1(t) \\
    \vdots \\
    x_N(t) \\
    \dot{x}_1(t) \\
    \vdots \\
    \dot{x}_N(t)
\end{pmatrix}
=
\begin{pmatrix}
    \sqrt{m_1} \dot{x}_1(t) \\
    \vdots \\
    \sqrt{m_N} \dot{x}_N(t) \\
    \sqrt{\kappa_{11}} x_1(t) \\
    \vdots \\
    \sqrt{\kappa_{NN}} x_N(t) \\
    \sqrt{\kappa_{12}} \left( x_1(t)-x_2(t) \right) \\
    \vdots \\
    \sqrt{\kappa_{N-1 \, N}} \left( x_{N-1}(t)-x_N(t) \right)
\end{pmatrix}.
\end{align}
Since we can check easily $B^{\prime \prime}AC_{B^{\prime \prime}}$ is a real anti-symmetric matrix using the fact $B^{\prime \prime}$ is a Koopman-von Neumann transforming matrix of $A$, $B^{\prime \prime}$ is a Koopman-von Neumann transforming matrix of $A$.
Now we focus on the relation between $B$ and $B^{\prime \prime}$, 
$B$ and $B^{\prime \prime }$ satisfy $B^{\prime \prime} =UB$, where
\begin{align}
U \coloneqq 
\begin{pmatrix}
    1 &  &  &  &  &  &  &  &  \\
      &\ddots &  &  &  &  &  &  &  \\
      &   &1 &  &  &  &  &  &  \\
      &  &  &-i &  &  &  &  &  \\
     &   &  &  &\ddots &  &  &  &  \\
     &   &  &  &  &-i &  &  &  \\
     &   &  &  &  &  &-i &  &  \\
     &   &  &  &  &  &  &\ddots &  \\
     &   &  &  &  &  &  &  &-i
\end{pmatrix}.
\end{align}
The mapped Hamiltonian via the Carleman embedding $H_{\rm Car}$ in Ref.~\cite{babbush_exponential_2023} is
\begin{align}
H_{\rm Car} = i BAC_B.
\end{align}
The mapped Hamiltonian via the Koopman-von Neumann embedding with $B^{\prime \prime}$ is
\begin{align}
H_{\rm Koo} = i B^{\prime \prime}AC_{B^{\prime \prime}}.
\end{align}
Then, the relation between $H_{\rm Car}$ and $H_{\rm Koo}$ is
\begin{align}
H_{\rm Koo} = U H_{\rm Car} U^\dagger.
\end{align}
Since $U$ has 1 and $-i$ regularly arranged on the diagonal elements, we can construct $U$ efficiently.
Then, there exists the unitary matrix $U$ that satisfies $H_{\rm Koo} = U H_{\rm Car} U^\dagger$ and can be constructed efficiently.
Therefore, we can solve the CHO problem efficiently, which includes the BQP-complete problem, by using Koopman-von Neumann embedding and a Hamiltonian simulation.

\section{Conclusion and Discussion} \label{section:conclusion}
In this study, we investigated the relation between Carleman and Koopman-von Neumann embedding, especially focusing on linear differential equations.
First, we investigated the necessary and sufficient condition for embedding 
a given linear differential equation into the Schr{\"o}dinger equation via Carleman embedding and Koopman-von Neumann embedding.
In both cases, the condition for reducing to the Schr\"{o}dinger equation is the same; $A$ is pure imaginary diagonalizable.
We also discussed the computational complexity of a related quantum algorithm to solve the linear differential equation using Carleman or Koopman-von Neumann embedding.
Next, we applied the above result to a special case where the linear differential equation 
are given as a Hamilton's canonical equation with a quadratic classical Hamiltonian.
We found that the matrix $\tilde H$ defining the classical quadratic Hamiltonian
is positive or negative, then the Hamilton's canonical equation can be mapped into 
the Schr{\"o}dinger equation.
The coupled harmonic oscillator system discussed in Ref.~\cite{babbush_exponential_2023}
is a special case of this framework.
Focusing on this relationship,
we showed the BQP-complete problem in Ref.~\cite{babbush_exponential_2023} belongs to the problem class that can be resolved efficiently using Koopman-von Neumann embedding.
This clearly indicates that 
solving differential equations using the Koopman-von Neumann embedding~\cite{engel_linear_2021,tanaka_quantum_2023}
includes nontrivial problems even in the linear case.
Furthermore, we found a nontrivial example that cannot be described as the coupled harmonic oscillator in Ref.~\cite{babbush_exponential_2023} but 
can be mapped to the Schr{\"o}dinger equation in our framework.

Our results have clarified the conditions under which general linear equations can be mapped to the Schr\"{o}dinger equation by means of Carleman or Koopman-von Neumann embedding.
This has laid the foundation for the concrete construction of future quantum algorithms for solving linear differential equations using Hamiltonian simulation. 
Furthermore, since the Carleman and Koopman-von Neumann embeddings are compatible in the linear domain, 
we expect that a perturbative approach of these arguments to the nonlinear domain will be useful for exploring quantum algorithms for nonlinear differential equations.

\begin{acknowledgements}
    This work is supported by MEXT Quantum Leap Flagship Program (MEXT Q-LEAP) Grant No. JPMXS0120319794 and JST COI-NEXT Grant No. JPMJPF2014.
\end{acknowledgements}

\bibliographystyle{apsrev4-1}
\bibliography{carleman-koopman-linear}

\begin{thebibliography}{20}%
\makeatletter
\providecommand \@ifxundefined [1]{%
 \@ifx{#1\undefined}
}%
\providecommand \@ifnum [1]{%
 \ifnum #1\expandafter \@firstoftwo
 \else \expandafter \@secondoftwo
 \fi
}%
\providecommand \@ifx [1]{%
 \ifx #1\expandafter \@firstoftwo
 \else \expandafter \@secondoftwo
 \fi
}%
\providecommand \natexlab [1]{#1}%
\providecommand \enquote  [1]{``#1''}%
\providecommand \bibnamefont  [1]{#1}%
\providecommand \bibfnamefont [1]{#1}%
\providecommand \citenamefont [1]{#1}%
\providecommand \href@noop [0]{\@secondoftwo}%
\providecommand \href [0]{\begingroup \@sanitize@url \@href}%
\providecommand \@href[1]{\@@startlink{#1}\@@href}%
\providecommand \@@href[1]{\endgroup#1\@@endlink}%
\providecommand \@sanitize@url [0]{\catcode `\\12\catcode `\$12\catcode
  `\&12\catcode `\#12\catcode `\^12\catcode `\_12\catcode `\%12\relax}%
\providecommand \@@startlink[1]{}%
\providecommand \@@endlink[0]{}%
\providecommand \url  [0]{\begingroup\@sanitize@url \@url }%
\providecommand \@url [1]{\endgroup\@href {#1}{\urlprefix }}%
\providecommand \urlprefix  [0]{URL }%
\providecommand \Eprint [0]{\href }%
\providecommand \doibase [0]{http://dx.doi.org/}%
\providecommand \selectlanguage [0]{\@gobble}%
\providecommand \bibinfo  [0]{\@secondoftwo}%
\providecommand \bibfield  [0]{\@secondoftwo}%
\providecommand \translation [1]{[#1]}%
\providecommand \BibitemOpen [0]{}%
\providecommand \bibitemStop [0]{}%
\providecommand \bibitemNoStop [0]{.\EOS\space}%
\providecommand \EOS [0]{\spacefactor3000\relax}%
\providecommand \BibitemShut  [1]{\csname bibitem#1\endcsname}%
\let\auto@bib@innerbib\@empty
\bibitem [{\citenamefont {Babbush}\ \emph {et~al.}(2023)\citenamefont
  {Babbush}, \citenamefont {Berry}, \citenamefont {Kothari}, \citenamefont
  {Somma},\ and\ \citenamefont {Wiebe}}]{babbush_exponential_2023}%
  \BibitemOpen
  \bibfield  {author} {\bibinfo {author} {\bibfnamefont {R.}~\bibnamefont
  {Babbush}}, \bibinfo {author} {\bibfnamefont {D.~W.}\ \bibnamefont {Berry}},
  \bibinfo {author} {\bibfnamefont {R.}~\bibnamefont {Kothari}}, \bibinfo
  {author} {\bibfnamefont {R.~D.}\ \bibnamefont {Somma}}, \ and\ \bibinfo
  {author} {\bibfnamefont {N.}~\bibnamefont {Wiebe}},\ }\href {\doibase
  10.48550/arXiv.2303.13012} {\enquote {\bibinfo {title} {Exponential quantum
  speedup in simulating coupled classical oscillators},}\ } (\bibinfo {year}
  {2023}),\ \bibinfo {note} {arXiv:2303.13012 [quant-ph]}\BibitemShut {NoStop}%
\bibitem [{\citenamefont {Shor}(1997)}]{shor_polynomial-time_1997}%
  \BibitemOpen
  \bibfield  {author} {\bibinfo {author} {\bibfnamefont {P.~W.}\ \bibnamefont
  {Shor}},\ }\href {\doibase 10.1137/S0097539795293172} {\bibfield  {journal}
  {\bibinfo  {journal} {SIAM Journal on Computing}\ }\textbf {\bibinfo {volume}
  {26}},\ \bibinfo {pages} {1484} (\bibinfo {year} {1997})},\ \bibinfo {note}
  {arXiv:quant-ph/9508027}\BibitemShut {NoStop}%
\bibitem [{\citenamefont {Reiher}\ \emph {et~al.}(2017)\citenamefont {Reiher},
  \citenamefont {Wiebe}, \citenamefont {Svore}, \citenamefont {Wecker},\ and\
  \citenamefont {Troyer}}]{reiher_elucidating_2017}%
  \BibitemOpen
  \bibfield  {author} {\bibinfo {author} {\bibfnamefont {M.}~\bibnamefont
  {Reiher}}, \bibinfo {author} {\bibfnamefont {N.}~\bibnamefont {Wiebe}},
  \bibinfo {author} {\bibfnamefont {K.~M.}\ \bibnamefont {Svore}}, \bibinfo
  {author} {\bibfnamefont {D.}~\bibnamefont {Wecker}}, \ and\ \bibinfo {author}
  {\bibfnamefont {M.}~\bibnamefont {Troyer}},\ }\href {\doibase
  10.1073/pnas.1619152114} {\bibfield  {journal} {\bibinfo  {journal}
  {Proceedings of the National Academy of Sciences}\ }\textbf {\bibinfo
  {volume} {114}},\ \bibinfo {pages} {7555} (\bibinfo {year} {2017})},\
  \bibinfo {note} {arXiv:1605.03590 [quant-ph]}\BibitemShut {NoStop}%
\bibitem [{\citenamefont {Harrow}\ \emph {et~al.}(2009)\citenamefont {Harrow},
  \citenamefont {Hassidim},\ and\ \citenamefont {Lloyd}}]{harrow_quantum_2009}%
  \BibitemOpen
  \bibfield  {author} {\bibinfo {author} {\bibfnamefont {A.~W.}\ \bibnamefont
  {Harrow}}, \bibinfo {author} {\bibfnamefont {A.}~\bibnamefont {Hassidim}}, \
  and\ \bibinfo {author} {\bibfnamefont {S.}~\bibnamefont {Lloyd}},\ }\href
  {\doibase 10.1103/PhysRevLett.103.150502} {\bibfield  {journal} {\bibinfo
  {journal} {Physical Review Letters}\ }\textbf {\bibinfo {volume} {103}},\
  \bibinfo {pages} {150502} (\bibinfo {year} {2009})},\ \bibinfo {note}
  {arXiv:0811.3171 [quant-ph]}\BibitemShut {NoStop}%
\bibitem [{\citenamefont {Berry}(2014)}]{berry_high-order_2014}%
  \BibitemOpen
  \bibfield  {author} {\bibinfo {author} {\bibfnamefont {D.~W.}\ \bibnamefont
  {Berry}},\ }\href {\doibase 10.1088/1751-8113/47/10/105301} {\bibfield
  {journal} {\bibinfo  {journal} {Journal of Physics A: Mathematical and
  Theoretical}\ }\textbf {\bibinfo {volume} {47}},\ \bibinfo {pages} {105301}
  (\bibinfo {year} {2014})},\ \bibinfo {note} {publisher: IOP
  Publishing}\BibitemShut {NoStop}%
\bibitem [{\citenamefont {Berry}\ \emph {et~al.}(2017)\citenamefont {Berry},
  \citenamefont {Childs}, \citenamefont {Ostrander},\ and\ \citenamefont
  {Wang}}]{berry_quantum_2017}%
  \BibitemOpen
  \bibfield  {author} {\bibinfo {author} {\bibfnamefont {D.~W.}\ \bibnamefont
  {Berry}}, \bibinfo {author} {\bibfnamefont {A.~M.}\ \bibnamefont {Childs}},
  \bibinfo {author} {\bibfnamefont {A.}~\bibnamefont {Ostrander}}, \ and\
  \bibinfo {author} {\bibfnamefont {G.}~\bibnamefont {Wang}},\ }\href {\doibase
  10.1007/s00220-017-3002-y} {\bibfield  {journal} {\bibinfo  {journal}
  {Communications in Mathematical Physics}\ }\textbf {\bibinfo {volume}
  {356}},\ \bibinfo {pages} {1057} (\bibinfo {year} {2017})}\BibitemShut
  {NoStop}%
\bibitem [{\citenamefont {Kowalski}(1997)}]{kowalski_nonlinear_1997}%
  \BibitemOpen
  \bibfield  {author} {\bibinfo {author} {\bibfnamefont {K.}~\bibnamefont
  {Kowalski}},\ }\href {\doibase 10.1063/1.531990} {\bibfield  {journal}
  {\bibinfo  {journal} {Journal of Mathematical Physics}\ }\textbf {\bibinfo
  {volume} {38}},\ \bibinfo {pages} {2483} (\bibinfo {year} {1997})},\ \bibinfo
  {note} {arXiv:solv-int/9801018}\BibitemShut {NoStop}%
\bibitem [{\citenamefont {Engel}\ \emph {et~al.}(2021)\citenamefont {Engel},
  \citenamefont {Smith},\ and\ \citenamefont {Parker}}]{engel_linear_2021}%
  \BibitemOpen
  \bibfield  {author} {\bibinfo {author} {\bibfnamefont {A.}~\bibnamefont
  {Engel}}, \bibinfo {author} {\bibfnamefont {G.}~\bibnamefont {Smith}}, \ and\
  \bibinfo {author} {\bibfnamefont {S.~E.}\ \bibnamefont {Parker}},\ }\href
  {\doibase 10.1063/5.0040313} {\bibfield  {journal} {\bibinfo  {journal}
  {Physics of Plasmas}\ }\textbf {\bibinfo {volume} {28}},\ \bibinfo {pages}
  {062305} (\bibinfo {year} {2021})},\ \bibinfo {note} {arXiv:2012.06681
  [physics, physics:quant-ph]}\BibitemShut {NoStop}%
\bibitem [{\citenamefont {Carleman}(1932)}]{carleman1932application}%
  \BibitemOpen
  \bibfield  {author} {\bibinfo {author} {\bibfnamefont {T.}~\bibnamefont
  {Carleman}},\ }\href@noop {} {\bibfield  {journal} {\bibinfo  {journal} {Acta
  Math}\ }\textbf {\bibinfo {volume} {59}},\ \bibinfo {pages} {63} (\bibinfo
  {year} {1932})}\BibitemShut {NoStop}%
\bibitem [{\citenamefont {Koopman}(1931)}]{koopman1931hamiltonian}%
  \BibitemOpen
  \bibfield  {author} {\bibinfo {author} {\bibfnamefont {B.~O.}\ \bibnamefont
  {Koopman}},\ }\href@noop {} {\bibfield  {journal} {\bibinfo  {journal}
  {Proceedings of the National Academy of Sciences}\ }\textbf {\bibinfo
  {volume} {17}},\ \bibinfo {pages} {315} (\bibinfo {year} {1931})}\BibitemShut
  {NoStop}%
\bibitem [{\citenamefont {v.~Neumann}(1932{\natexlab{a}})}]{von-Neumann_1932a}%
  \BibitemOpen
  \bibfield  {author} {\bibinfo {author} {\bibfnamefont {J.}~\bibnamefont
  {v.~Neumann}},\ }\href {http://www.jstor.org/stable/1968537} {\bibfield
  {journal} {\bibinfo  {journal} {Annals of Mathematics}\ }\textbf {\bibinfo
  {volume} {33}},\ \bibinfo {pages} {587} (\bibinfo {year}
  {1932}{\natexlab{a}})}\BibitemShut {NoStop}%
\bibitem [{\citenamefont {v.~Neumann}(1932{\natexlab{b}})}]{von-Neumann_1932b}%
  \BibitemOpen
  \bibfield  {author} {\bibinfo {author} {\bibfnamefont {J.}~\bibnamefont
  {v.~Neumann}},\ }\href {http://www.jstor.org/stable/1968225} {\bibfield
  {journal} {\bibinfo  {journal} {Annals of Mathematics}\ }\textbf {\bibinfo
  {volume} {33}},\ \bibinfo {pages} {789} (\bibinfo {year}
  {1932}{\natexlab{b}})}\BibitemShut {NoStop}%
\bibitem [{\citenamefont {Liu}\ \emph {et~al.}(2021)\citenamefont {Liu},
  \citenamefont {Kolden}, \citenamefont {Krovi}, \citenamefont {Loureiro},
  \citenamefont {Trivisa},\ and\ \citenamefont {Childs}}]{liu_efficient_2021}%
  \BibitemOpen
  \bibfield  {author} {\bibinfo {author} {\bibfnamefont {J.-P.}\ \bibnamefont
  {Liu}}, \bibinfo {author} {\bibfnamefont {H.~{\O}.}\ \bibnamefont {Kolden}},
  \bibinfo {author} {\bibfnamefont {H.~K.}\ \bibnamefont {Krovi}}, \bibinfo
  {author} {\bibfnamefont {N.~F.}\ \bibnamefont {Loureiro}}, \bibinfo {author}
  {\bibfnamefont {K.}~\bibnamefont {Trivisa}}, \ and\ \bibinfo {author}
  {\bibfnamefont {A.~M.}\ \bibnamefont {Childs}},\ }\href {\doibase
  10.1073/pnas.2026805118} {\bibfield  {journal} {\bibinfo  {journal}
  {Proceedings of the National Academy of Sciences}\ }\textbf {\bibinfo
  {volume} {118}},\ \bibinfo {pages} {e2026805118} (\bibinfo {year} {2021})},\
  \bibinfo {note} {arXiv:2011.03185 [physics, physics:quant-ph]}\BibitemShut
  {NoStop}%
\bibitem [{\citenamefont {Childs}\ \emph {et~al.}(2017)\citenamefont {Childs},
  \citenamefont {Kothari},\ and\ \citenamefont {Somma}}]{childs_quantum_2017}%
  \BibitemOpen
  \bibfield  {author} {\bibinfo {author} {\bibfnamefont {A.~M.}\ \bibnamefont
  {Childs}}, \bibinfo {author} {\bibfnamefont {R.}~\bibnamefont {Kothari}}, \
  and\ \bibinfo {author} {\bibfnamefont {R.~D.}\ \bibnamefont {Somma}},\ }\href
  {\doibase 10.1137/16M1087072} {\bibfield  {journal} {\bibinfo  {journal}
  {SIAM Journal on Computing}\ }\textbf {\bibinfo {volume} {46}},\ \bibinfo
  {pages} {1920} (\bibinfo {year} {2017})},\ \bibinfo {note} {arXiv:1511.02306
  [quant-ph]}\BibitemShut {NoStop}%
\bibitem [{\citenamefont {Tanaka}\ and\ \citenamefont
  {Fujii}(2023)}]{tanaka_quantum_2023}%
  \BibitemOpen
  \bibfield  {author} {\bibinfo {author} {\bibfnamefont {Y.}~\bibnamefont
  {Tanaka}}\ and\ \bibinfo {author} {\bibfnamefont {K.}~\bibnamefont {Fujii}},\
  }\href {\doibase 10.48550/arXiv.2305.00653} {\enquote {\bibinfo {title}
  {Quantum {Solvable} {Nonlinear} {Differential} {Equations}},}\ } (\bibinfo
  {year} {2023}),\ \bibinfo {note} {arXiv:2305.00653 [quant-ph]}\BibitemShut
  {NoStop}%
\bibitem [{\citenamefont {Low}\ and\ \citenamefont
  {Chuang}(2019)}]{low_hamiltonian_2019}%
  \BibitemOpen
  \bibfield  {author} {\bibinfo {author} {\bibfnamefont {G.~H.}\ \bibnamefont
  {Low}}\ and\ \bibinfo {author} {\bibfnamefont {I.~L.}\ \bibnamefont
  {Chuang}},\ }\href {\doibase 10.22331/q-2019-07-12-163} {\bibfield  {journal}
  {\bibinfo  {journal} {Quantum}\ }\textbf {\bibinfo {volume} {3}},\ \bibinfo
  {pages} {163} (\bibinfo {year} {2019})},\ \bibinfo {note} {arXiv:1610.06546
  [quant-ph]}\BibitemShut {NoStop}%
\bibitem [{\citenamefont {Gily{\'e}n}\ \emph {et~al.}(2019)\citenamefont
  {Gily{\'e}n}, \citenamefont {Su}, \citenamefont {Low},\ and\ \citenamefont
  {Wiebe}}]{gilyen_quantum_2019}%
  \BibitemOpen
  \bibfield  {author} {\bibinfo {author} {\bibfnamefont {A.}~\bibnamefont
  {Gily{\'e}n}}, \bibinfo {author} {\bibfnamefont {Y.}~\bibnamefont {Su}},
  \bibinfo {author} {\bibfnamefont {G.~H.}\ \bibnamefont {Low}}, \ and\
  \bibinfo {author} {\bibfnamefont {N.}~\bibnamefont {Wiebe}},\ }in\ \href
  {\doibase 10.1145/3313276.3316366} {\emph {\bibinfo {booktitle} {Proceedings
  of the 51st {Annual} {ACM} {SIGACT} {Symposium} on {Theory} of
  {Computing}}}}\ (\bibinfo {year} {2019})\ pp.\ \bibinfo {pages} {193--204},\
  \bibinfo {note} {arXiv:1806.01838 [quant-ph]}\BibitemShut {NoStop}%
\bibitem [{\citenamefont {Knill}\ \emph {et~al.}(2007)\citenamefont {Knill},
  \citenamefont {Ortiz},\ and\ \citenamefont {Somma}}]{knill_optimal_2007}%
  \BibitemOpen
  \bibfield  {author} {\bibinfo {author} {\bibfnamefont {E.}~\bibnamefont
  {Knill}}, \bibinfo {author} {\bibfnamefont {G.}~\bibnamefont {Ortiz}}, \ and\
  \bibinfo {author} {\bibfnamefont {R.~D.}\ \bibnamefont {Somma}},\ }\href
  {\doibase 10.1103/PhysRevA.75.012328} {\bibfield  {journal} {\bibinfo
  {journal} {Physical Review A}\ }\textbf {\bibinfo {volume} {75}},\ \bibinfo
  {pages} {012328} (\bibinfo {year} {2007})},\ \bibinfo {note} {publisher:
  American Physical Society}\BibitemShut {NoStop}%
\bibitem [{\citenamefont {Forets}\ and\ \citenamefont
  {Pouly}(2017)}]{forets_explicit_2017}%
  \BibitemOpen
  \bibfield  {author} {\bibinfo {author} {\bibfnamefont {M.}~\bibnamefont
  {Forets}}\ and\ \bibinfo {author} {\bibfnamefont {A.}~\bibnamefont {Pouly}},\
  }\href {\doibase 10.48550/arXiv.1711.02552} {\enquote {\bibinfo {title}
  {Explicit {Error} {Bounds} for {Carleman} {Linearization}},}\ } (\bibinfo
  {year} {2017}),\ \bibinfo {note} {arXiv:1711.02552 [cs, math]}\BibitemShut
  {NoStop}%
\bibitem [{\citenamefont {Chihara}(1978)}]{chihara-intro-orthogonal-poly}%
  \BibitemOpen
  \bibfield  {author} {\bibinfo {author} {\bibfnamefont {T.~S.}\ \bibnamefont
  {Chihara}},\ }\href@noop {} {\emph {\bibinfo {title} {An Introduction to
  Orthogonal Polynomials}}}\ (\bibinfo  {publisher} {Gordon and Breach, Science
  Publishers},\ \bibinfo {year} {1978})\BibitemShut {NoStop}%
\end{thebibliography}%

\appendix
\section{Carleman Embedding} \label{appendix:carleman-embedding}
Following Refs.~\cite{liu_efficient_2021, forets_explicit_2017}, we introduce the Carleman embedding which is the method of transforming a nonlinear differential equation into a linear differential equation.
First,
we prepare the tensor product of $x=(x_1, \dots, x_N)^\top \in \mathbb{R}^N$
\begin{equation}
    x^{[j]} \coloneqq \underbrace{x \otimes \dots \otimes x}_j .
\end{equation}
For example,
when $j=1, 2$, $x^{[1]}, x^{[2]}$ are $x^{[1]}=x, \ x^{[2]}=(x_1{}^2, x_1 x_2, \dots, x_{N} x_{N-1}, x_N {}^2)^\top$.
Assume that the differential equation is given by
\begin{equation} \label{eq:Nd-poly-differential-eq}
    \frac{d}{dt}
    \begin{pmatrix}
        x_1 \\
        \vdots \\
        x_N
    \end{pmatrix}
    =
    \begin{pmatrix}
        p_1(x_1, \dots, x_N) \\
        \vdots \\
        p_n(x_1, \dots, x_N)
    \end{pmatrix},
\end{equation}
where $p_1(X_1, \dots, X_N), \dots, p_n(X_1, \dots, X_N)$ are real polynomials in variables $X_1, \dots, X_N$ 
$\ \mathrm{with} \ p_1(\bm{0})=\dots=p_n(\bm{0})=0$.
Let $k$ be $k \coloneqq \max_{j \in \{1, \dots, N\} } \deg p_j(X_1, \dots, X_N)$,
Eq. \eqref{eq:Nd-poly-differential-eq} is rewritten by
\begin{equation} \label{eq:Nd-poly-differential-eq-matrix}
    \frac{d}{dt}
    \begin{pmatrix}
        x_1 \\
        \vdots \\
        x_N
    \end{pmatrix}
    =F_1 x+ F_2 x^{[2]}+\dots +F_k x^{[k]},
\end{equation}
where $F_l \in \mathbb{R}^{N \times N^l} \ (l \in \{1, \dots, k\})$ are suitable matrices.
Now, we define the transfer matrix
\begin{equation}
    A_{j+l-1}^j \coloneqq \sum_{\nu =1}^j \overbrace{I_N \otimes \dots \otimes \underbrace{F_l}_\nu \otimes \dots \otimes I_N}^j 
    \in \mathbb{R}^{N^j \times N^{j+l-1}} \hspace{1em} (j \in \mathbb{Z}_{>0}, l \in \{1, \dots, k\}),
\end{equation}
then we can transform Eq. \eqref{eq:Nd-poly-differential-eq-matrix} into
\begin{equation}
    \frac{d}{dt} x^{[j]} = \sum_{l=0}^{k-1} A_{j+l}^j x^{[j+l]} \hspace{1em} (j \in \mathbb{Z}_{>0}).
\end{equation}
Therefore,
we obtain
\begin{equation} \label{eq:Carleman-definition}
    \frac{d}{dt}
    \begin{pmatrix}
        x^{[1]} \\
        x^{[2]} \\
        x^{[3]} \\
        \vdots
    \end{pmatrix}
    =
    \begin{pmatrix}
        A_1^1 &A_2^1 &A_3^1 &\hdots &A_k^1 &0 &0 &\hdots \\
        0 &A_2^2 &A_3^2 &\hdots &A_k^2 &A_{k+1}^2 &0 &\hdots \\
        0 &0 &A_3^3 &\hdots &A_k^3 &A_{k+1}^3 &A_{k+2}^3 &\hdots \\
        \vdots &\vdots &\vdots & &\vdots &\vdots &\vdots &
    \end{pmatrix}
        \begin{pmatrix}
        x^{[1]} \\
        x^{[2]} \\
        x^{[3]} \\
        \vdots
    \end{pmatrix}.
\end{equation}
Let the matrix in the right hand of Eq. \eqref{eq:Carleman-definition} be $\mathcal{A}$,
and each components of $x^{[j]}$ be explicitly,
then we obtain
\begin{equation} \label{eq:Carleman-definition-2}
    \frac{d}{dt}
    \begin{pmatrix}
        x_1 \\
        \vdots \\
        x_n \\
        x_1 {}^2 \\
        x_1 x_2 \\
        \vdots \\
        x_n x_{n-1} \\
        x_n {}^2 \\
        \vdots
    \end{pmatrix}
    = \mathcal{A}
    \begin{pmatrix}
        x_1 \\
        \vdots \\
        x_n \\
        x_1 {}^2 \\
        x_1 x_2 \\
        \vdots \\
        x_n x_{n-1} \\
        x_n {}^2 \\
        \vdots
    \end{pmatrix}.
\end{equation}
The above method is called Carleman embedding, which transforms a nonlinear differential equation \eqref{eq:Nd-poly-differential-eq} into the linear differential equation of the form like Eq. \eqref{eq:Carleman-definition} and \eqref{eq:Carleman-definition-2}.

\section{Koopman-von Neumann Embedding} \label{appendix:koopman}
Following Refs.~\cite{kowalski_nonlinear_1997, tanaka_quantum_2023}, we introduce the Koopman-von Neumann embedding which is the method of transforming a nonlinear differential equation into a linear differential equation.
First, we introduce an orthonormal polynomial sequence and the Rodrigues' Formula.
Then, we show the details of the Koopman-von Neumann embedding with the Hermite polynomial.
We note that the Hermite polynomial $\{H_n(x)\}_{n \in \mathbb{Z}_{\ge 0}}$ is the orthonormal polynomial sequence and satisfies the Rodrigues' formula \eqref{eq:def-Rodrigues}.

\subsection{Rodrigues' Formula} \label{appendix-sub:Rodrigues}
First, we define a weight function and an orthonormal polynomial sequence as described in Ref.~\cite{chihara-intro-orthogonal-poly} with some modification.

\begin{definition}
    Let $a,b \in \mathbb{R}\cup \{\pm \infty\} \ \mathrm{with} \ a<b$.
    A function $w \colon (a,b) \to \mathbb{R}_{\ge 0}$ is called a weight function,
    when $w$ satisfies $w(x)>0$ on almost everywhere $x \in (a,b)$ and $\int_a^b \lvert x \rvert^n w(x) \, dx < \infty \ (n \in \mathbb{Z}_{\ge 0})$.
\end{definition}

\begin{definition}
    Let $a,b \in \mathbb{R}\cup \{\pm \infty\} \ \mathrm{with} \ a<b$, and $w$ be a weight function on $(a,b)$.
    A sequence $\{p_n(x)\}_{n=0}^\infty \subset \mathbb{R}[x]$ is called an orthonormal polynomial sequence,
    when $\{p_n(x)\}_{n=0}^\infty$ satisfies the following two conditions,
    \begin{enumerate}
        \item $p_n(x)$ is a polynomial of degree $n$ for $n \in \mathbb{Z}_{\ge 0}$,
        \item $\int_a^b p_m(x)p_n(x) w(x) \, dx = \delta_{m,n}$ for $m,n \in \mathbb{Z}_{\ge 0}$.
    \end{enumerate}
\end{definition}

Next, we define the Rodrigues' formula that defines a typical orthonormal polynomial sequence as described in Ref.~\cite{chihara-intro-orthogonal-poly} with some modification.

\begin{definition}
    Let $a,b$ be in $\mathbb{R}\cup \{\pm \infty\} \ \mathrm{with} \ a<b$, and $w$ be the weight function on $(a,b)$ that satisfies $w^\prime (x)/w(x) = (cx+d)/ X(x)$, where $c,d \in \mathbb{R}$ and $X(x) \in \mathbb{R}[x] \ \mathrm{with} \deg X(x) \le 2$.
    Moreover, we assume that $w(x), X(x)$ also satisfy $\left( \frac{d}{dx} \right)^k \left[ w(x)X(x)^n \right] =0$ at $x=a,b$ for $0 \le k < n$.
    An orthonormal polynomial sequence $\{p_n(x)\}_{n=0}^\infty$ is defined by a formula of type
    \begin{equation} \label{eq:def-Rodrigues}
        p_n(x) = \frac{1}{K_n} \frac{1}{w(x)} \frac{d^n}{dx^n} \left[ w(x) X(x)^n \right] \hspace{1em} (n \in \mathbb{Z}_{\ge 0}),
    \end{equation}
    where $K_n \in \mathbb{R} \setminus \{0\}$ is an appropriate constant.
    Eq. \eqref{eq:def-Rodrigues} is called the Rodrigues' formula.
\end{definition}

\subsection{Koopman-von Neumann Embedding with the Hermite Polynomial} \label{appendix-sub:Koopman-von Neumann-hermite-polynomial}
Now we select the Hermite polynomial $\{H_n(x)\}_{n \in \mathbb{Z}_{\ge 0}}$ as the orthonormal polynomial sequence that satisfies the Rodrigues' formula \eqref{eq:def-Rodrigues} in Appendix \ref{appendix-sub:Rodrigues}.
We assume that the following differential equations are given,
\begin{equation} \label{eq:Nd-nonlinear-differential-eq}
    \frac{d}{dt} x_j = F_j (\bm{x}) \hspace{1em} (j \in \{1, \dots, N \}),
\end{equation}
where $\bm{x}=(x_1, \dots, x_N)^\top$, and analytic function $F \colon \mathbb{R}^N \ni \bm{x} \mapsto \left( F_1(\bm{x}), \dots, F_N(\bm{x}) \right)^\top \in \mathbb{R}^N$.
Let $a_j^\dagger, a_j \ (j \in \{1, \dots, N\})$ be $j$th bosonic creation and annihilation operators, and $\{ \ket{n_j} \}_{n_j=0}^\infty$ be an eigenvector sequence of the number operator $N_j = a_j^\dagger a_j$ such that $N_j \ket{n_j}=n_j \ket{n_j}$.
Now we take the Hilbert space spanned by $\{ \otimes_{j=1}^N \ket{n_j} \}_{n_1, \dots, n_N \in \mathbb{Z}_{\ge 0}}$.
We introduce the special vector $\ket{\bm{x}}=\otimes_{j=1}^N \ket{x_j} \ \left(x_1, \dots, x_N \in (a,b) \right)$ defined by 
\begin{equation}
    \ket{\bm{x}}
    \coloneqq \bigotimes_{j=1}^N \left( w(x_j)^{1/2} \sum_{n_j=0}^\infty H_{n_j} (x_j) \ket{n_j} \right).
\end{equation}
We define the hermitian operator $\hat{x}_j \ (j \in \{1, \dots, N\})$ that is represented by $a_j^\dagger, a_j$ and satisfies $\hat{x}_j \ket{\bm{x}}= x_j \ket{\bm{x}}$ for $x_1, \dots, x_N \in (a,b)$.
We also define the hermitian operator $\hat{k}_j \ (j \in \{1, \dots, N\})$ that is represented by $a_j^\dagger, a_j$ and satisfies $[ \hat{x}_j, \hat{k}_l ] = i \delta_{j, l}$ for $j,l \in \{1, \dots, N\}$.
Now, the solution $\bm{x}$ of Eq. \eqref{eq:Nd-nonlinear-differential-eq} satisfies the differential equation
\begin{equation} \label{eq:koopman-definition}
    i \frac{d}{dt} \tilde{\ket{\bm{x}}} = \hat{H} \tilde{\ket{\bm{x}}},
\end{equation}
where the operator $\hat{H}$ is
\begin{equation} \label{eq:hamiltonian-koopman-hermite-poly}
    \hat{H}=\sum_{j=1}^N \frac{1}{2} \left( \hat{k}_j F_j(\hat{x}_j) +  F_j(\hat{x}_j) \hat{k}_j \right),
\end{equation}
and the vector $\tilde{\ket{\bm{x}}}$ is
\begin{equation} \label{eq:koopman-x-Nd}
\tilde{\ket{\bm{x}}} = \exp \left( \frac{1}{2} \int_0^t \mathrm{div} \bm{F} (\bm{x}(\tau)) \, d \tau \right) \ket{\bm{x}}.
\end{equation}
To obtain the Schr\"{o}dinger equation of $\ket{\bm{x}}$, we assume that \begin{equation} \label{eq:div-0}
    \mathrm{div} F(\bm{x}) = 0
\end{equation}
for $\bm{x} \in (a,b)^N$.
From Eq. \eqref{eq:koopman-x-Nd}, we obtain $\tilde{\ket{\bm{x}}} = \ket{\bm{x}}$,
then Eq. \eqref{eq:koopman-definition} becomes 
\begin{equation} \label{eq:koopman-definition-tr0}
    i \frac{d}{dt} \ket{\bm{x}} = \hat{H} \ket{\bm{x}}.
\end{equation}
The above method is called Koopman-von Neumann embedding, which transforms a nonlinear differential equation \eqref{eq:Nd-nonlinear-differential-eq} into the linear differential equation of the form like Eq. \eqref{eq:koopman-definition}.
In this case, we note that we choose the Hermite polynomial as the orthonormal polynomial sequence that satisfies the Rodrigues' formula \eqref{eq:def-Rodrigues} in Appendix \ref{appendix-sub:Rodrigues}.

\section{The Proof of Theorems in Section \ref{section:carleman-linear}, \ref{section:Koopman-von Neumann} and \ref{section:relation-carleman-koopman}}
\subsection{The Proof of Theorem \ref{thm:carleman-schrodinger}} \label{appendix-sub:carleman-thm}
First, we characterize an anti-hermitian matrix.

\begin{lemma} \label{lem:anti-hermitian}
For $N \in \mathbb{Z}_{>0}$,
let $A \in \mathbb{C}^{N \times N}$,
the following conditions are equivalent:
\begin{enumerate}
    \item $A$ is an anti-hermitian matrix.
    \item For all $\ket{x},\ket{y} \in \mathbb{C}^N$, $\braket{x | (A+A^\dagger) | y} = 0$.
\end{enumerate}
\end{lemma}

\begin{proof}
    $(1. \Longrightarrow 2.)$
    Since $A$ is an anti-hermitian matrix, $A+A^\dagger = 0$. Then, for all $\ket{x},\ket{y} \in \mathbb{C}^N$, $\braket{x | (A+A^\dagger) | y} = 0$.
    \newline
    $(2. \Longrightarrow 1.)$
    Let $\ket{e_j} \in \mathbb{C}^N \ (j \in \{1, \dots, N \} )$ be the vector whose $j$ th component is 1 and the others are 0.
    From condition 2., for all $j,k \in \{1, \dots, N\}$, $\braket{e_j | (A+A^\dagger) | e_k}=0$, then $A+A^\dagger=0$. Thus, $A$ is an anti-hermitian matrix.
\end{proof}

Next, we characterize the condition of Theorem \ref{thm:carleman-schrodinger} using $C_B$.

\begin{prop} \label{prop:BACB-anti-hermitian}
Let $A \in \mathbb{R}^{N \times N}$.
For $B \in \mathbb{C}^{M \times N} \ \mathrm{with} \ \mathrm{rank} B =N$,
the following conditions are equivalent:
\begin{enumerate}
    \item There exists $C \in \mathbb{C}^{N \times M} \  \mathrm{with} \ CB=I_N$ such that $BAC$ is an anti-hermitian matrix.
    \item $BAC_B$ is an anti-hermitian matrix.
\end{enumerate}
\end{prop}

\begin{proof}
    The implication $(2. \Longrightarrow 1.)$ is quite clear. To prove the converse, let's assume that condition 1. holds, and we want to show that $BAC_B = BAC$.
    First, notice that $C_B B = C B = I_N$. This means that for any vector $\ket{x} \in \mathrm{Im} B$, we have $C_B \ket{x} = C \ket{x}$. Consequently, we have $BAC_B \ket{x} = BAC \ket{x}$ for $\ket{x} \in \mathrm{Im} B$.
    
    Since $BAC$ is an anti-hermitian matrix, we can apply Lemma \ref{lem:anti-hermitian}. Therefore, for all pairs of vectors $(\ket{x}, \ket{y}) \in (\mathrm{Im} B) \times (\mathrm{Im} B)^\perp$, we have $\braket{x| BAC + (BAC)^\dagger |y} = 0$.
    Now, because $\ket{y} \in (\mathrm{Im} B)^\perp$, we can conclude that $\braket{x|(BAC)^\dagger |y} = (BAC \ket{x})^\dagger \ket{y} = 0$. Consequently, $\braket{x| BAC|y} = 0$.
    We also have $\braket{z| BAC|y} = 0$ for any $\ket{z} \in (\mathrm{Im} B)^\perp$.
    Given that $\mathbb{C}^M = (\mathrm{Im} B) \oplus (\mathrm{Im} B)^\perp$, for all pairs of vectors $(\ket{x}, \ket{y}) \in \mathbb{C}^M \times (\mathrm{Im} B)^\perp$, we still have $\braket{x| BAC|y} = 0$. 
    Therefore, $BAC \ket{y} = 0$ for all $\ket{y} \in (\mathrm{Im} B)^\perp$.
    According to Lemma \ref{lem:property-CB}, it is evident that $C_B |_{(\mathrm{Im} B)^\perp} = 0$. Consequently, for all $\ket{y} \in (\mathrm{Im} B)^\perp$, we have $BAC_B \ket{y} = 0$.
    Then, $BAC_B \ket{y} = BAC \ket{y}$ for all $\ket{y} \in (\mathrm{Im} B)^\perp$.
    
    Since $\mathbb{C}^M = (\mathrm{Im} B) \oplus (\mathrm{Im} B)^\perp$, for all vectors $\ket{x} \in \mathbb{C}^M$, we have $BAC_B \ket{x} = BAC \ket{x}$. Therefore, $BAC_B = BAC$.
    Considering condition 1., we can conclude that $BAC$ is an anti-hermitian matrix, which implies that $BAC_B$ is also an anti-hermitian matrix.
    As a result, condition 2. is satisfied.
\end{proof}

From the proof of Proposition \ref{prop:BACB-anti-hermitian}, the following Corollary \ref{cor:BAC=BACB} holds.

\begin{cor} \label{cor:BAC=BACB}
    Under the setting of Proposition \ref{prop:BACB-anti-hermitian}, if $BAC$ is an anti-hermitian matrix, then $BAC_B=BAC$.
\end{cor}

Next,
we investigate the condition equivalent to $BAC_B$ being an anti-hermitian matrix.

\begin{prop} \label{prop:BdBA-anti-hermitian}
Let $A \in \mathbb{R}^{N \times N}$.
For $B \in \mathbb{C}^{M \times N} \ \mathrm{with} \ \mathrm{rank} B =N$,
the following conditions are equivalent:
\begin{enumerate}
    \item $BAC_B$ is an anti-hermitian matrix.
    \item $B^\dagger BA$ is an anti-hermitian matrix.
\end{enumerate}
\end{prop}

\begin{proof}
    $(1. \Longrightarrow 2.)$
    Since $BAC_B$ is an anti-hermitian matrix and $C_BB=I_N$ from Lemma \ref{lem:property-CB}, $B^\dagger (BAC_B) B=B^\dagger BA$ is an anti-hermitian matrix.
    Therefore, condition 2. holds.
    \newline
    $(2. \Longrightarrow 1.)$
    Since $\mathrm{rank} B =N$, $B^\dagger B$ is a regular matrix as proved in Lemma \ref{lem:property-CB}.
    From condition 2., 
    \begin{align}
        (B^\dagger B)^{-1} (B^\dagger BA) \left( (B^\dagger B)^{-1} \right)^\dagger = A (B^\dagger B)^{-1}
    \end{align} 
    is an anti-hermitian matrix.
    Since $C_B = (B^\dagger B)^{-1} B^\dagger$,
    $B\left( A (B^\dagger B)^{-1} \right) B^\dagger = BAC_B$ is an anti-hermitian matrix.
    Therefore, condition 1. holds.
\end{proof}

Next,
we prove the equivalence of the set of $B^\dagger B$ and the set of positive matrices.

\begin{lemma} \label{lem:complex-positive}
\begin{equation}
\left\{ B^\dagger B \mid B \in \mathbb{C}^{M \times N} \ \mathrm{with} \ \mathrm{rank} B =N \right\}
=\left\{ D \in \mathbb{C}^{N \times N} \mid D >0 \right\}
\end{equation}
holds.
\end{lemma}

\begin{proof}
    $( \subset )$
    Let $B \in \mathbb{C}^{M \times N} \ \mathrm{with} \ \mathrm{rank} B =N$.
    Since $B^\dagger B$ satisfies $B^\dagger B \ge 0$ and $B^\dagger B$ is a regular matrix as proved in Lemma \ref{lem:property-CB}, thus $B^\dagger B$ satisfies $B^\dagger B > 0$.
    That is, $B^\dagger B \in \left\{ D \in \mathbb{C}^{N \times N} \mid D >0 \right\}$.
    \newline
    $( \supset )$
    Let $D \in \mathbb{C}^{N \times N} \ \mathrm{with} \  D >0$.
    Since $D>0$, $\sqrt{D} \in \mathbb{C}^{N \times N}$ exists and satisfies $\sqrt{D}>0$.
    $\sqrt{D}>0$ leads $\mathrm{rank} \sqrt{D} = N$.
    We suppose $B \coloneqq 
    \begin{pmatrix}
        \sqrt{D} \\
        O_{(M-N) \times N}
    \end{pmatrix}
    \in \mathbb{C}^{M \times N}
    $.
    Since $\mathrm{rank} \sqrt{D} = N$, $\mathrm{rank} B = N$ holds.
    From the definition of $B$, $B^\dagger B = D$ holds.
    Then, $D \in \left\{ B^\dagger B \mid B \in \mathbb{C}^{M \times N} \ \mathrm{with} \ \mathrm{rank} B =N \right\}$.
\end{proof}

Proposition \ref{prop:DA-anti-hermitian} is proved by Lemma \ref{lem:complex-positive}.

\begin{prop} \label{prop:DA-anti-hermitian}
For $A \in \mathbb{R}^{N \times N}$,
the following conditions are equivalent:
\begin{enumerate}
    \item There exists $B \in \mathbb{C}^{M \times N} \ \mathrm{with} \ \mathrm{rank} B=N$ such that $B^\dagger BA$ is an anti-hermitian matrix.
    \item There exists $D \in \mathbb{C}^{N \times N} \ \mathrm{with} \ D>0$ such that $DA$ is an anti-hermitian matrix.
\end{enumerate}
\end{prop}

Finally, we show the equivalence of the condition about only $A$.

\begin{prop} \label{prop:complex-A-diagonalizable-pure-imaginary}
For $A \in \mathbb{R}^{N \times N}$,
the following conditions are equivalent:
\begin{enumerate}
    \item There exists $D \in \mathbb{C}^{N \times N} \ \mathrm{with} \ D>0$ such that $DA$ is an anti-hermitian matrix.
    \item $A$ is diagonalizable and $\sigma_p(A) \subset i \mathbb{R}$.
\end{enumerate}
\end{prop}

\begin{proof}
    $(1. \Longrightarrow 2.)$
    Since $D>0$, $\sqrt{D} \in \mathbb{C}^{N \times N}$ exists and satisfies $\sqrt{D}>0$.
    $DA$ is an anti-hermitian matrix, then $\sqrt{D}^{-1} (DA) \left( \sqrt{D}^{-1} \right)^\dagger = \sqrt{D} A \sqrt{D}^{-1}$ is an anti-hermitian matrix.
    Since $A$ is similar to $\sqrt{D} A \sqrt{D}^{-1}$ and $\sqrt{D} A \sqrt{D}^{-1}$ is an anti-hermitian matrix, $A$ is diagonalizable and $\sigma_p (A) = \sigma_p \left( \sqrt{D} A \sqrt{D}^{-1} \right) \subset i \mathbb{R}$.
    Therefore, condition 2. holds.
    \newline
    $(2. \Longrightarrow 1.)$
    According to condition 2., there exists $P \in \mathbb{C}^{N \times N}$ such that $PAP^{-1}$ is a diagonal matrix.
    Since $\sigma_p (A) \subset i \mathbb{R}$, $PAP^{-1}$ is an anti-hermitian matrix.
    Since $P$ is a regular matrix, we can apply the polar decomposition for $P$, then there exists a unitary matrix $U \in \mathbb{C}^{N \times N}$ and a positive matrix $P' \in \mathbb{C}^{N \times N}$ (i.e., $P'>0$) such that $P=UP'$.
    Since $PAP^{-1}$ is an anti-hermitian matrix, $U^\dagger (PAP^{-1}) U = U^\dagger (UP'AP'^{-1}U^\dagger) U = P'AP'^{-1}$ is also an anti-hermitian matrix.
    Since $P'AP'^{-1}$ is an anti-hermitian matrix, $P' (P'AP'^{-1}) P'^\dagger = P'^2A$ is also an anti-hermitian matrix.
    $P'$ satisfies $P' > 0$, then $P'^2 > 0$ holds.
    Let $D \in \mathbb{C}^{N \times N}$ be $D=P'^2$, $D$ satisfies $D>0$ and $DA$ is an anti-hermitian matrix.
    Therefore, condition 1. holds.
\end{proof}

From Proposition \ref{prop:BACB-anti-hermitian}, \ref{prop:BdBA-anti-hermitian}, \ref{prop:DA-anti-hermitian} and \ref{prop:complex-A-diagonalizable-pure-imaginary},
Theorem \ref{thm:carleman-schrodinger} is proved.

\subsection{The Proof of Theorem \ref{thm:koopman-linear-schrodinger-condition}} \label{appendix-sub:koopman}
First, we characterize an anti-symmetric matrix.

\begin{lemma} \label{lem:anti-symmetric}
For $N \in \mathbb{Z}_{>0}$,
let $A \in \mathbb{R}^{N \times N}$, 
the following conditions are equivalent:
\begin{enumerate}
    \item $A$ is a real anti-symmetric matrix.
    \item For all $\ket{x} \in \mathbb{R}^N$, $\braket{x|A|x} = 0$.
\end{enumerate}
\end{lemma}

\begin{proof}
    $(1. \Longrightarrow 2.)$
    $A$ is a real anti-symmetric matrix, then $A+A^\top = 0$.
    Then, for all $\ket{x} \in \mathbb{R}^N$, $0 = \braket{x|\left(A+A^\top\right)|x}=2 \braket{x|A|x}$, thus $\braket{x|A|x} = 0$.
    Therefore, condition 2. holds.
    \newline
    $(2. \Longrightarrow 1.)$
    From condition 2., for all $\ket{x}, \ket{y} \in \mathbb{R}^N$, 
    \begin{equation}
        \braket{x| (A+A^\top) | y} 
        = \frac{1}{2} \left\{ 
        (\bra{x}+\bra{y}) A (\ket{x}+\ket{y}) 
        -  (\bra{x}-\bra{y}) A (\ket{x}-\ket{y})
        \right\}
        =0.
    \end{equation}
    Then $A+A^\top = 0$ holds, that is, $A \in \mathbb{R}^{N \times N}$ is a real anti-symmetric matrix.
\end{proof}

Next, using $C_B$, we characterize the condition of Theorem \ref{thm:koopman-linear-schrodinger-condition}.

\begin{prop} \label{prop:koopman-preserving}
    We suppose that a differential equation $\frac{d}{dt} \bm{x}(t) = A \bm{x} (t)$ with any initial value $\bm{x} (0) \in \mathbb{R}^N$, where $A \in \mathbb{R}^{N \times N}$.
    The following conditions are equivalent:
    \begin{enumerate}
        \item There exists $(B, C) \in \mathbb{R}^{M \times N} \times \mathbb{R}^{N \times M} \ \mathrm{with} \ \mathrm{rank} B = N \ \mathrm{and} \ CB=I_N$ such that 
        a linear differential equation for $\bm{y}(t):=B\bm{x}$, 
        $\frac{d}{dt} \bm{y}(t)  = BAC \bm{y}(t)$ satisfies the mapped Hamiltonian preserves the total fock number and $\mathrm{div} (BAC \bm{y})=0$, i.e., $\mathrm{tr} (BAC) =0$ holds.
        \item There exists $B \in \mathbb{R}^{M \times N} \ \mathrm{with} \ \mathrm{rank} B = N$ such that $BAC_B$ is a real anti-symmetric matrix.
    \end{enumerate}
\end{prop}

\begin{proof}
    $(1. \Longrightarrow 2.)$
    According to condition 1., the Hamiltonian preserves the total fock number, so the Hamiltonian can be block diagonalized with bases $\{ 
    \ket{0 \dots 0}, \ket{1 0 \dots 0}, \dots, \ket{0 \dots 0 1}, \ket{2 0 \dots 0}, \ket{110\dots 0}, \dots, \ket{0 \dots 02}, \dots \}$.
    Since $\mathrm{div} (BAC \bm{y}) =0$, 0-th Hermite polynomial is $H_0(x)=1$, 1-st Hermite polynomial is $H_1(x)=2x$, and the weight function of the Hermite polynomials is $w(x)=\exp(-x^2)$, the Schr\"{o}dinger equation like Eq. \eqref{eq:koopman-definition-tr0} restricted in the subspace spanned by $\{\ket{10\dots0}, \dots, \ket{0\dots01} \}$ becomes
    \begin{equation} \label{eq:koopman-schrodinger-1}
        i \frac{d}{dt}
        \left\{
        \exp \left( - \sum_{j=1}^M \frac{y_j(t)^2}{2} \right)
        \begin{pmatrix}
        2 y_1(t) \\
        \vdots \\
        2 y_M(t)
        \end{pmatrix}
        \right\}
        =
        \hat{H}^{(1)}
        \left\{
        \exp \left( - \sum_{j=1}^M \frac{y_j(t)^2}{2} \right)
        \begin{pmatrix}
        2 y_1(t) \\
        \vdots \\
        2 y_M(t)
        \end{pmatrix}
        \right\},
    \end{equation}
    where $\bm{y}(t) \coloneqq (y_1(t), \dots, y_M(t))^\top$ and $\hat{H}^{(1)}$ is the Hamiltonian matrix restricted in the subspace spanned by $\{\ket{10\dots0}, \dots, \ket{0\dots01} \}$.
    Since the norm of the state vector in a Schr\"{o}dinger equation is preserved, so
    \begin{equation}
        \left\lVert 
        \exp \left(- \sum_{j=1}^M \frac{y_j(t)^2}{2} \right)
        \begin{pmatrix}
        2 y_1(t) \\
        \vdots \\
        2 y_M(t)
        \end{pmatrix}
        \right\rVert
        =
        2 \exp \left(- \sum_{j=1}^M \frac{y_j(t)^2}{2} \right) \sqrt{ \sum_{j=1}^M y_j(t)^2 }
    \end{equation}
    is constant.
    By the continuity of $\bm{y}(t)$ and the number of the solution $x$ of equation $e^{-x/2} \sqrt{x}=a \ (a \ge 0)$ being finite, $\sum_{j=1}^M y_j(t)^2$ is constant.
    Thus, we have transformed Eq. \eqref{eq:koopman-schrodinger-1} into the following form, denoted as Eq. \eqref{eq:koopman-schrodinger-2}:
    \begin{equation} \label{eq:koopman-schrodinger-2}
        \frac{d}{dt}
        \begin{pmatrix}
            y_1(t) \\
            \vdots \\
            y_M(t)
        \end{pmatrix}
        =
        -i \hat{H}^{(1)}
        \begin{pmatrix}
            y_1(t) \\
            \vdots \\
            y_M(t)
        \end{pmatrix}.
    \end{equation}
    Since $\bm{y}(t) = B \bm{x}(t)$, we can rewrite the equation as:
    \begin{equation} \label{eq:normalized-B-H1}
        \frac{d}{dt}
        B
        \begin{pmatrix}
            x_1(t) \\
            \vdots \\
            x_N(t)
        \end{pmatrix}
        =
        -i \hat{H}^{(1)}
        B
        \begin{pmatrix}
            x_1(t) \\
            \vdots \\
            x_N(t)
        \end{pmatrix},
    \end{equation}
    where $\bm{x} (t) \coloneqq (x_1(t), \dots, x_N(t))^\top$.
    Additionally, considering the equation $\frac{d}{dt} \bm{x}(t) = A \bm{x}(t)$, we can derive:
    \begin{equation} \label{eq:A-to-BA}
        \frac{d}{dt}
        B
        \begin{pmatrix}
            x_1(t) \\
            \vdots \\
            x_N(t)
        \end{pmatrix}
        =
        BA
        \begin{pmatrix}
            x_1(t) \\
            \vdots \\
            x_N(t)
        \end{pmatrix}.
    \end{equation}
    From Eq. \eqref{eq:normalized-B-H1} and \eqref{eq:A-to-BA}, we obtain the equation:
    \begin{equation} \label{eq:BA-H1B}
        -i \hat{H}^{(1)} B
        \begin{pmatrix}
            x_1(t) \\
            \vdots \\
            x_N(t)
        \end{pmatrix}
        =
        BA
        \begin{pmatrix}
            x_1(t) \\
            \vdots \\
            x_N(t)
        \end{pmatrix}.
    \end{equation}
    By the continuity of $(x_1(t), \dots, x_N(t))^\top$, $t \downarrow 0$ leads  Eq. \eqref{eq:BA-H1B} holds at $t=0$.
    Eq. \eqref{eq:BA-H1B} at $t=0$ holds for any initial value $\bm{x} (0) \in \mathbb{R}^N$, then $BA=-i \hat{H}^{(1)}B$.
    Thus $BAC_B = -i \hat{H}^{(1)}B C_B$ holds.
    Since $\mathbb{R}^M$ can be represented as $\mathbb{R}^M = (\mathrm{Im} B) \oplus (\mathrm{Im} B)^\perp$, let $\ket{y}$ be $\ket{y}=\ket{y_1}+\ket{y_2}$, where $(\ket{y_1}, \ket{y_2}) \in (\mathrm{Im} B) \times (\mathrm{Im} B)^\perp$.
    $C_B B = I_N$ by Lemma \ref{lem:property-CB} leads $BC_B |_{\mathrm{Im} B} = \mathrm{Id}_{\mathrm{Im} B}$.
    Since $C_B |_{(\mathrm{Im} B)^\perp}=0$, $BC_B |_{\mathrm{Im} B} = \mathrm{Id}_{\mathrm{Im} B}$, and $\hat{H}^{(1)}$ is a hermitian matrix, then $\braket{y | BAC_B | y}$ becomes
    \begin{align}
        \braket{y | BAC_B | y} 
        &= (\bra{y_1} + \bra{y_2}) BAC_B (\ket{y_1} + \ket{y_2}) \\
        &= \braket{y_1 | BAC_B | y_1} \\
        &= \braket{y_1 | (-i \hat{H}^{(1)}B C_B) | y_1} \\
        &=-i \braket{y_1 | \hat{H}^{(1)} | y_1} \in i \mathbb{R}.
    \end{align}
    $BAC_B \in \mathbb{R}^{M \times M}$ and $\ket{y} \in \mathbb{R}^M$ leads $\braket{y | BAC_B | y} \in \mathbb{R}$, then $\braket{y | BAC_B | y} \in \mathbb{R} \cap i \mathbb{R} = \{0\}$, thus $\braket{y | BAC_B | y} = 0$ holds for all $\ket{y} \in \mathbb{R}^M$.
    According to Lemma \ref{lem:anti-symmetric}, $BAC_B$ is an anti-symmetric matrix.
    Therefore, condition 2. holds.
    \newline
    $(2. \Longrightarrow 1.)$
    First, we show $\mathrm{tr} (BAC_B)=0$.
    Since $BAC_B$ is a real anti-symmetric matrix, $\mathrm{tr} (BAC_B)=0$ holds.
    
    Next, let $BAC_B$ be $BAC_B=(\tilde{a}_{jk})_{j,k=1,\dots, M}$.
    Since $BAC_B$ is an anti-symmetric matrix, $\tilde{a}_{kj}=-\tilde{a}_{jk} \ (j, k \in \{1, \dots, M\})$ holds.
    Now, let's rephrase Eq. \eqref{eq:hamiltonian-koopman-hermite-poly} using the operators $\hat{x}_j = (a_j^\dagger + a_j) / \sqrt{2}$ and $\hat{k}_j= i (a_j^\dagger - a_j) / \sqrt{2} \ (j \in \{1, \dots, M\})$, 
    where $a_j^\dagger, a_j \ (j \in \{1, \dots, N\})$ are $j$th bosonic creation and annihilation operators in Appendix \ref{appendix-sub:Koopman-von Neumann-hermite-polynomial}. 
    Eq. \eqref{eq:hamiltonian-koopman-hermite-poly} can be rewritten as follows:
    \begin{align}
        \hat{H} &= \frac{1}{2}
        \left\{
        \sum_{j=1}^M \frac{i(a_j^\dagger-a_j)}{\sqrt{2}} \left(\sum_{k=1}^M \tilde{a}_{jk} \frac{a_k^\dagger+a_k}{\sqrt{2}} \right)
        +
        \sum_{j=1}^M  \left(\sum_{k=1}^M \tilde{a}_{jk} \frac{a_k^\dagger+a_k}{\sqrt{2}} \right) \frac{i(a_j^\dagger-a_j)}{\sqrt{2}}
        \right\} \\
        &= \frac{1}{2}
        \left\{
        \sum_{j=1}^M \sum_{k=1}^M  \tilde{a}_{jk} \frac{i(a_j^\dagger-a_j)(a_k^\dagger+a_k)}{2} 
        -
        \sum_{k=1}^M \sum_{j=1}^M  \tilde{a}_{kj} \frac{i(a_k^\dagger+a_k)(a_j^\dagger-a_j)}{2} 
        \right\} \\
        &=\frac{i}{2} \sum_{j=1}^M \sum_{k=1}^M \tilde{a}_{jk} \left(a^\dagger_j a_k -a_j a^\dagger_k \right). \label{eq:hamiltonian-representation}
    \end{align}
    This representation implies the Hamiltonian preserves the total fock number.
    Therefore, condition 1. holds.
\end{proof}

From Proposition \ref{prop:koopman-preserving},
the following Corollary \ref{cor:koopman-to-carleman} holds.

\begin{cor} \label{cor:koopman-to-carleman}
    Under the setting of Proposition \ref{prop:koopman-preserving},
    if $B \in \mathbb{R}^{M \times N}$ is a Koopman-von Neumann transforming matrix of $A$, then $B$ is a Carleman transforming matrix of $A$.
\end{cor}

\begin{proof}
    Since $B \in \mathbb{R}^{M \times N}$ is a Koopman-von Neumann transforming matrix of $A$, $BAC_B$ is a real anti-symmetric matrix according to Proposition \ref{prop:koopman-preserving}.
    Since a real anti-symmetric matrix is an anti-hermitian matrix, $BAC_B$ is an anti-hermitian matrix.
    According to Proposition \ref{prop:BACB-anti-hermitian} and the definition of the Carleman transforming matrix, $B$ is a Carleman transforming matrix of $A$.
\end{proof}

Focus on the matrix representation of $\hat{H}$ with $\{\ket{10\dots0}, \dots, \ket{0\dots01} \}$,
the following Corollary \ref{cor:koopman-hamiltonian-BACB} holds.

\begin{cor} \label{cor:koopman-hamiltonian-BACB}
    Under the setting of Proposition \ref{prop:koopman-preserving},
    the matrix representation of $-i\hat{H}$ with $\{\ket{10\dots0}, \dots, \ket{0\dots01} \}$ is coincidence with $BAC_B$ by using Eq. \eqref{eq:hamiltonian-representation}.
\end{cor}

\begin{proof}
    We can check the relation by calculation.
\end{proof}

Next,
we investigate the condition equivalent to $BAC_B$ being a real anti-symmetric matrix.

\begin{prop} \label{prop:tBBA-anti-symmetric}
Let $A \in \mathbb{R}^{N \times N}$.
For $B \in \mathbb{R}^{M \times N} \ \mathrm{with} \ \mathrm{rank} B =N$,
the following conditions are equivalent:
\begin{enumerate}
    \item $BAC_B$ is a real anti-symmetric matrix.
    \item $B^\top BA$ is a real anti-symmetric matrix.
\end{enumerate}
\end{prop}

\begin{proof}
    $(1. \Longrightarrow 2.)$
    Since $BAC_B$ is a real anti-symmetric matrix and $C_BB=I_N$ from Lemma \ref{lem:property-CB}, $B^\top (BAC_B) B=B^\top BA$ is an anti-symmetric matrix.
    By $A \in \mathbb{R}^{N \times N}$ and $B \in \mathbb{R}^{M \times N}$, $B^\top BA \in \mathbb{R}^{N \times N}$ is a real anti-symmetric matrix.
    Therefore, condition 2. holds.
    \newline
    $(2. \Longrightarrow 1.)$
    Since $B \in \mathbb{R}^{M \times N} \ \mathrm{with} \ \mathrm{rank} B =N$, $B^\dagger B = B^\top B$ is a regular matrix as proved in Lemma \ref{lem:property-CB}.
    From condition 2., $(B^\top B)^{-1} (B^\top BA) \left( (B^\top B)^{-1} \right)^\top = A (B^\top B)^{-1}$ is an anti-symmetric matrix.
    Since $C_B = (B^\dagger B)^{-1} B^\dagger$, $B\left( A (B^\top B)^{-1} \right) B^\top = BA \left( (B^\dagger B)^{-1} B^\dagger \right) = BAC_B$ is an anti-symmetric matrix.
    Therefore, condition 1. holds.
\end{proof}

Next,
we prove the equivalence of the set of $B^\top B$ and the set of real positive matrices.

\begin{lemma} \label{lem:real-positive}
\begin{equation}
\left\{ B^\top B \mid B \in \mathbb{R}^{M \times N} \ \mathrm{with} \ \mathrm{rank} B =N \right\}
=\left\{ D \in \mathbb{R}^{N \times N} \mid D >0 \right\}
\end{equation}
holds.
\end{lemma}

\begin{proof}
    $( \subset )$
    Let $B \in \mathbb{R}^{M \times N} \ \mathrm{with} \ \mathrm{rank} B =N$.
    Since $B^\top B$ satisfies $B^\top B \ge 0$ and $B^\dagger B = B^\top B$ is a regular matrix as proved in Lemma \ref{lem:property-CB}, $B^\top B$ satisfies $B^\top B > 0$.
    That is, $B^\top B \in \left\{ D \in \mathbb{R}^{N \times N} \mid D >0 \right\}$.
    \newline
    $( \supset )$
    Let $D \in \mathbb{R}^{N \times N} \ \mathrm{with} \  D >0$.
    Since $D>0$ and $D$ is a real matrix, $\sqrt{D} \in \mathbb{R}^{N \times N}$ exists and satisfies $\sqrt{D}>0$.
    $\sqrt{D}>0$ leads $\mathrm{rank} \sqrt{D} = N$.
    We suppose $B \coloneqq 
    \begin{pmatrix}
        \sqrt{D} \\
        O_{(M-N) \times N}
    \end{pmatrix}
    \in \mathbb{R}^{M \times N}
    $.
    Since $\mathrm{rank} \sqrt{D} = N$, $\mathrm{rank} B = N$ holds.
    From the definition of $B$, $B^\top B = D$ holds.
    Then, $D \in \left\{ B^\top B \mid B \in \mathbb{R}^{M \times N} \ \mathrm{with} \ \mathrm{rank} B =N \right\}$.
\end{proof}

Proposition \ref{prop:DA-anti-symmetric} is proved by Lemma \ref{lem:real-positive}.

\begin{prop} \label{prop:DA-anti-symmetric}
For $A \in \mathbb{R}^{N \times N}$, the following conditions are equivalent:
\begin{enumerate}
    \item There exists $B \in \mathbb{R}^{M \times N} \ \mathrm{with} \ \mathrm{rank} B=N$ such that $B^\top BA$ is a real anti-symmetric matrix.
    \item There exists $D \in \mathbb{R}^{N \times N} \ \mathrm{with} \ D>0$ such that $DA$ is a real anti-symmetric matrix.
\end{enumerate}
\end{prop}

To use the proof of Theorem \ref{thm:carleman-schrodinger} for showing Theorem \ref{thm:koopman-linear-schrodinger-condition}, we prove the following Lemma \ref{lem:Re-BdB}.

\begin{lemma} \label{lem:Re-BdB}
    Let $A \in \mathbb{R}^{N \times N}$, and $B \in \mathbb{C}^{M \times N} \ \mathrm{with} \ \mathrm{rank} B =N$.
    If $B^\dagger BA$ is an anti-hermitian matrix, then $(\mathfrak{Re}(B^\dagger B) )A$ is a real anti-symmetric matrix.
\end{lemma}

\begin{proof}
    Since $B^\dagger BA$ is an anti-hermitian matrix, $\mathfrak{Re} (B^\dagger BA)$ is a real anti-symmetric matrix.
    $A$ is a real matrix, then $\mathfrak{Re} (B^\dagger BA) = (\mathfrak{Re}(B^\dagger B) )A$.
    Thus, $(\mathfrak{Re}(B^\dagger B) )A$ is a real anti-symmetric matrix.
\end{proof}

Finally, we show the equivalence of the condition about only $A$.

\begin{prop} \label{prop:real-A-diagonalizable-pure-imaginary}
For $A \in \mathbb{R}^{N \times N}$, the following conditions are equivalent:
\begin{enumerate}
    \item There exists $D \in \mathbb{R}^{N \times N} \ \mathrm{with} \ D>0$ such that $DA$ is a real anti-symmetric matrix.
    \item $A$ is diagonalizable and $\sigma_p(A) \subset i \mathbb{R}$.
\end{enumerate}
\end{prop}

\begin{proof}
    $(1. \Longrightarrow 2.)$
    Since $D>0$ and $D$ is a real matrix, $\sqrt{D} \in \mathbb{R}^{N \times N}$ exists and satisfies $\sqrt{D}>0$.
    $DA$ is a real anti-symmetric matrix, then $\sqrt{D}^{-1} (DA) \left( \sqrt{D}^{-1} \right)^\top = \sqrt{D} A \sqrt{D}^{-1} \in \mathbb{R}^{N \times N}$ is a real anti-symmetric matrix.
    Since $A$ is similar to $\sqrt{D} A \sqrt{D}^{-1}$ and $\sqrt{D} A \sqrt{D}^{-1}$ is a real anti-symmetric matrix, $A$ is diagonalizable and $\sigma_p (A) = \sigma_p \left( \sqrt{D} A \sqrt{D}^{-1} \right) \subset i \mathbb{R}$.
    Therefore, condition 2. holds.
    \newline
    $(2. \Longrightarrow 1.)$
    According to Proposition \ref{prop:DA-anti-hermitian} and \ref{prop:complex-A-diagonalizable-pure-imaginary}, there exists $B \in \mathbb{C}^{M \times N} \ \mathrm{with} \ \mathrm{rank} B =N$ such that $B^\dagger BA$ is an anti-hermitian matrix.
    We apply Lemma \ref{lem:Re-BdB} for $B^\dagger BA$, then $(\mathfrak{Re}(B^\dagger B) )A$ is a real anti-symmetric matrix.
    Let $D=\mathfrak{Re}(B^\dagger B) \in \mathbb{R}^{N \times N}$, then $DA$ is a real anti-symmetric matrix.
    According to Lemma \ref{lem:complex-positive}, $B^\dagger B$ satisfies $B^\dagger B > 0$, then $\braket{x | B^\dagger B |x} > 0$ holds for all $\ket{x} \in \mathbb{R}^N \setminus \{0\}$.
    Since $\ket{x} \in \mathbb{R}^N$ and $\braket{x | B^\dagger B |x} > 0$, $\mathfrak{Re} \left( \braket{x | B^\dagger B |x} \right) = \braket{x|D|x} >0$ holds.
    Thus, $D$ satisfies $D>0$.
    Therefore, condition 1. holds.
\end{proof}

From Proposition \ref{prop:koopman-preserving}, \ref{prop:tBBA-anti-symmetric}, \ref{prop:DA-anti-symmetric} and \ref{prop:real-A-diagonalizable-pure-imaginary},
Theorem \ref{thm:koopman-linear-schrodinger-condition} is proved.

\subsection{The Proof of Theorem \ref{thm:carleman-to-koopman}} \label{appendix-sub:real-imaginary}
\textit{Proof of Theorem \ref{thm:carleman-to-koopman}.}
From the definition of the Carleman transforming matrix and Proposition \ref{prop:BACB-anti-hermitian} and \ref{prop:BdBA-anti-hermitian}, $B^\dagger BA$ is an anti-hermitian matrix.
According to Lemma \ref{lem:Re-BdB}, $(\mathfrak{Re}(B^\dagger B) )A$ is a real anti-symmetric matrix.
Since $B^{\prime \top} B^\prime A = (\mathfrak{Re}(B^\dagger B) )A$, $B^{\prime \top} B^\prime A$ is a real anti-symmetric matrix.
Let $\ket{x} \in \mathrm{Ker} B^\prime$, $\ket{x}$ satisfies $(\mathfrak{Re} B) \ket{x} = (\mathfrak{Im} B) \ket{x}=0$.
Thus $0 = (\mathfrak{Re} B + i \mathfrak{Im} B) \ket{x} = B \ket{x}$, then $\ket{x} \in \mathrm{Ker} B$.
$\mathrm{rank} B = N$ and rank–nullity theorem lead $\mathrm{Ker} B =\{0\}$, then $\ket{x}=0$.
This implies $\mathrm{Ker} B^\prime = \{0\}$.
Rank–nullity theorem leads $\mathrm{rank} B^\prime = N$.
That is, $B^\prime \in \mathbb{R}^{2M \times N}$ satisfies $B^{\prime \top} B^\prime A$ is a real anti-symmetric matrix and $\mathrm{rank} B^\prime = N$, so we can apply Proposition \ref{prop:tBBA-anti-symmetric} for $B^\prime$. 
Then, $B^\prime AC_{B^\prime}$ is a real anti-symmetric matrix.
According to the definition of the Koopman-von Neumann transforming matrix, $B^\prime$ is a Koopman-von Neumann transforming matrix of $A$.
\qed

\end{document}